\documentclass[journal,10pt,doublecolumn]{IEEEtran}

\usepackage{amsmath}
\usepackage{graphicx}
\usepackage{cite}
\usepackage{amsfonts}
\usepackage{amsthm}
\usepackage{algorithm}
\usepackage{algorithmic}
\usepackage{amssymb}
\usepackage{color}
\usepackage{setspace}

\theoremstyle{definition} \newtheorem{lemma}{Lemma}

\begin{document}

\title{Securing UAV Communications via Joint Trajectory and Power Control}

\author{Guangchi~Zhang,~\IEEEmembership{Member,~IEEE,}
Qingqing~Wu,~\IEEEmembership{Member,~IEEE,}
Miao~Cui,
Rui~Zhang,~\IEEEmembership{Fellow,~IEEE}
\thanks{G. Zhang and M. Cui are with the School of Information Engineering, Guangdong University of Technology, Guangzhou, China (email: \{gczhang, cuimiao\}@gdut.edu.cn). Q. Wu and R. Zhang are with the Department of Electrical and Computer Engineering, National University of Singapore (email: \{elewuqq, elezhang\}@nus.edu.sg). Q. Wu is the corresponding author. Part of this paper was presented in IEEE Global Communications Conference (GLOBECOM), Singapore, Dec. 2017 \cite{Zhang2017GC}.

This work was supported in part by the National Natural Science Foundation of China under Grant 61571138, in part by the Science and Technology Plan Project of Guangdong Province under Grants 2017B090909006 and 2016B090904001, and in part by the Science and Technology Plan Project of Guangzhou City under Grant 201803030028.}}

\maketitle

\begin{abstract}
Unmanned aerial vehicle (UAV) communication is anticipated to be widely applied in the forthcoming fifth-generation (5G) wireless networks, due to its many advantages such as low cost, high mobility, and on-demand deployment. However, the broadcast and line-of-sight (LoS) nature of air-to-ground wireless channels gives rise to a new challenge on how to realize secure UAV communications with the destined nodes on the ground. This paper aims to tackle this challenge by applying the physical layer security technique. We consider both the downlink and uplink UAV communications with a ground node, namely UAV-to-ground (U2G) and ground-to-UAV (G2U) communications, respectively, subject to a potential eavesdropper on the ground. In contrast to the existing literature on wireless physical layer security only with ground nodes at fixed or quasi-static locations, we exploit the high mobility of the UAV to proactively establish favorable and degraded channels for the legitimate and eavesdropping links, respectively, via its trajectory design. We formulate new problems to maximize the average secrecy rates of the U2G and G2U transmissions, respectively, by jointly optimizing the UAV's trajectory and the transmit power of the legitimate transmitter over a given flight period of the UAV. Although the formulated problems are non-convex, we propose iterative algorithms to solve them efficiently by applying the block coordinate descent and successive convex optimization methods. Specifically, the transmit power and UAV trajectory are each optimized with the other being fixed in an alternating manner, until the algorithms converge. Simulation results show that the proposed algorithms can improve the secrecy rates for both U2G and G2U communications, as compared to other benchmark schemes without power control and/or trajectory optimization.
\end{abstract}


\IEEEpeerreviewmaketitle

\section{Introduction}
With many advantages such as high mobility, low cost, wide coverage, and on-demand deployment, unmanned aerial vehicles (UAVs) have been extensively used in both military and civilian applications, such as search and rescue, inspection and surveillance, cargo transportation, etc. Recently, UAVs have also found increasingly more substantial applications in wireless communication \cite{Zeng2016}, and are expected to play a significant role in the forthcoming fifth-generation (5G) wireless networks \cite{wu:magazine,WuSurvey}. To seize this growing opportunity, internationally leading telecommunication companies such as Qualcomm, Ericsson, and China Mobile have already launched their research projects on integrating UAVs into the 5G networks \cite{Qualcomm2016, Ericsson2016}. Generally speaking, there are two main paradigms of UAV applications in 5G. In the first one, termed as ``UAV-assisted wireless communication'', UAVs are utilized as airbone communication platforms such as mobile base stations (BSs) and/or relays that can be flexibly deployed on demand to assist the communications in terrestrial networks such as 5G.  For example, UAV-mounted BSs can be used to enable rapid wireless communication service recovery after ground infrastructure damages, or provide offloading service for terrestrial BSs in extremely crowded areas \cite{Mozaffari2015, Mozaffari2016, Bor2016, Lyu2017, WuGC20170, WuGC2017, WuGC20172,WuGC20173,Sharma2016}. Another example is to use UAVs as mobile relays to provide reliable connectivity between distant users in remote areas (e.g., an uninhabited desert) that are not covered by any existing wireless networks \cite{Zeng2016Trans, Johansen2014}. Moreover, in future internet of things (IoT) applications, UAVs can be dispatched to disseminate/collect data to/from widespread distributed wireless devices efficiently and with low cost \cite{Sotheara2014, Lyu2016, Mozaffari2017}. By contrast, in the other paradigm, known as ``cellular-enabled UAV communication'', UAVs are regarded as new ``sky'' users in the cellular networks that enable two-way communications of the UAVs with ground BSs. For example, the future 5G networks can provide reliable communications for UAVs even beyond the range of their operators' visual line-of-sight (LoS) to achieve long-range UAV control in real time \cite{Qualcomm2016b}. Besides, in UAV-enabled surveillance applications, the captured pictures and/or videos by the UAVs in real time can be uploaded timely to the ground data centers via the 5G networks \cite{Motlagh2017}.

In the aforementioned UAV communication applications in 5G, due to the broadcast nature of wireless channels, their security and privacy are of utmost concern \cite{Gopala2008, Liang2008}. One major advantage of UAV-ground communications is that UAVs usually have LoS channels for the communications with ground nodes, especially in outdoor environments. However, such LoS communication links are also more prone to the eavesdropping by illegitimate nodes on the ground, which gives rise to a new security challenge. Although security was conventionally viewed as a higher layer communication protocol stack design problem that can be tackled by using cryptographic methods, physical layer security has emerged as a promising alternative way of defense to realize secrecy in wireless communication.

A key design metric that has been widely adopted in physical layer security is the so-called secrecy rate \cite{Gopala2008, Liang2008, Wang2011, Xing2016, Khisti2010, Zheng2011, Zou2013, Anli2008, Zhang2016, ZhangTVT2018, Cui2018, Tang2016, Tang2018}, at which confidential message can be reliably transmitted without having the eavesdropper infer any information about the message. A non-zero secrecy rate can be achieved when the strength of the legitimate link is stronger than that of the eavesdropping link. In the existing literature on physical layer security, communication nodes are usually assumed to be at fixed or quasi-static locations. As a result, the average channel quality of the legitimate/eavesdropping link mainly depends on the path loss and shadowing from the transmitter to receiver, which are determined if the locations of the legitimate transmitter/receiver and the eavesdropper are given. Thus, in the case that the average channel gain of the legitimate receiver is smaller than that of the eavesdropper (e.g., due to longer distance from the legitimate transmitter), in order to achieve positive secrecy rates, the  exploitation of the wireless channel small-scale fading in time, frequency, and/or space becomes essential, and various techniques such as power control in time and/or frequency as well as multi-antenna beamforming have been investigated. In \cite{Gopala2008}, power control with rate adaptation over fading channels is proposed to maximize the average secrecy rate. This work is also extended to characterize the secrecy rate region of parallel-fading broadcast channels \cite{Liang2008}. In \cite{Wang2011}, power control over frequency subcarriers is investigated for secrecy rate maximization in an orthogonal frequency-division multiple access (OFDMA) system. In \cite{Xing2016}, joint power control on information signal and artificial noise (AN) is proposed to maximize the secrecy rate of a simultaneous wireless information and power transfer (SWIPT) system. In multiple-input multiple-output (MIMO) systems, transmit beamforming can be jointly employed with AN transmission to effectively enhance the legitimate link capacity and at the same time degrade that of the eavesdropping link. For example, the legitimate transmitter can use beamforming to steer a null to the eavesdropper, or send AN in the direction of the eavesdropper to interferer with it \cite{Khisti2010}. In \cite{Tang2016}, beamforming is jointly designed with channel coding to achieve unconditional security in MIMO communications. If one or more relay helpers are available, they can also cooperatively send AN or jamming signals to interfere with the eavesdroppers to achieve better secrecy communication performance. In \cite{Zheng2011}, optimal cooperative jamming via relays is studied to maximize the secrecy rate of a single-antenna point-to-point legitimate link. Besides, transmission scheduling by exploiting multiuser channel diversity is another effective approach to improve the secrecy communication performance in a system with multiple legitimate users/eavesdroppers. In \cite{Zou2013}, a transmission scheduling scheme is proposed to maximize the secrecy rate of a multiuser cognitive radio network. In \cite{Tang2018}, it is shown that for a terrestrial point-to-point wireless communication system, a moving receiver can achieve better secrecy performance than that of the system with a static receiver.

However, there are still two major challenges that remain unsolved in the existing physical layer security literature. First, the practically achievable secrecy rate can be severely limited if the distance between the legitimate transmitter and its intended receiver is fixed and significantly larger than that between it and a potential eavesdropper, even if the various approaches mentioned above are applied. Second, the channel state information (CSI) of the eavesdropper is usually required at the legitimate transmitter for the implementation of effective power control and/or beamforming techniques. This is practically challenging since the eavesdropper is usually a passive device and thus it is difficult to estimate such CSI. In this paper, we study physical layer security in UAV-ground communications, which may potentially overcome the above two critical issues in conventional studies. First, in contrast to the existing literature with fixed or quasi-static nodes only, the high mobility of UAVs can be exploited to proactively establish stronger links with the legitimate ground nodes and/or degrade the channels of the eavesdroppers, by flying closer/farther to/from them, respectively, via proper trajectory design. This approach is particularly effective in the context of UAV-ground communications (as compared to conventional terrestrial communications), since the LoS links are usually much more dominant over other channel impairments such as shadowing and small-scale fading, due to the much larger height of the UAV than typical ground nodes such as mobile terminals or BSs. Furthermore, since the LoS channel gain is only determined by the link distance, the UAV can practically obtain the channel gain to any potential eavesdropper on the ground if its location is known, which thus resolves the eavesdropper-CSI issue in the existing literature. Note that the location of any ground node as a potential eavesdropper can be practically detected and tracked by the UAV via using an optical camera or synthetic aperture radar (SAR) equipped on the UAV \cite{UAVDisaster, UAVSAR}.

\begin{figure}[!t]
	\centering
	\includegraphics[width=\columnwidth]{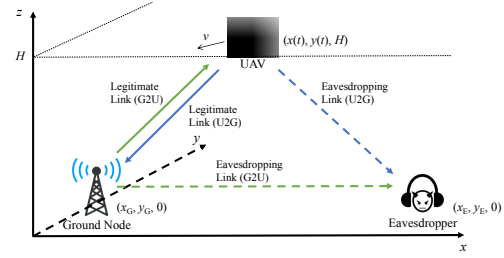}
	\caption{A UAV wireless communication system consisting of a UAV above ground and a node on the ground. A potential eavesdropper on the ground intends to intercept the wireless communication between them.}  \label{FigDLULModel}
\end{figure}

For an initial exposition, in this paper we consider a simplified three-node secrecy UAV-ground communication system as shown in Fig. \ref{FigDLULModel}, where a UAV at fixed altitude intends to communicate with a ground node, while a potential eavesdropper on the ground may intercept their communication. The secure communications of both UAV-to-ground (U2G) and ground-to-UAV (G2U) links are considered. In the U2G case, the UAV and the ground node are the legitimate transmitter and receiver, respectively, where both the legitimate and eavesdropping links are modeled as LoS channels. By contrast, in the G2U case, the ground node and the UAV are the legitimate transmitter and receiver, respectively. Since the legitimate transmitter and potential eavesdropper are both on the ground in this case, different from the U2G case, only the legitimate link is modeled as a LoS channel, while the eavesdropping link is practically modeled as a channel consisting of both distance-dependent path-loss and small-scale Rayleigh fading. Thus, the problem formulations for the secrecy rate maximization in these two cases are generally different, which will be detailed later in this paper. Nevertheless, the secrecy rates of both U2G and G2U transmissions can benefit from the joint design of UAV trajectory and transmit power control at the legitimate transmitter (i.e., UAV and ground node in the U2G and G2U cases, respectively), explained as follows. On one hand, the location of the UAV can be adjusted dynamically to establish stronger channels for the legitimate link than that for the eavesdropping link. On the other hand, due to practical constraints such as the UAV's initial and final locations, the legitimate link may not be always stronger than the eavesdropping link during the whole flight period of the UAV. In this case, transmit power can be adapted to the channel variations arising from the UAV's movement to further improve the secrecy rate. For example, in the U2G case, the UAV should transmit higher power when it flies closer to the ground node while being more far away from the eavesdropper, and transmit lower or zero power otherwise.

Motivated by this, we aim to design joint UAV trajectory and transmit power optimization algorithms to secure both U2G and G2U communications. Our goal is to maximize the average secrecy rate over a finite flight period of the UAV in each case, subject to the practical mobility constraints on the UAV's maximum speed and its initial and final locations, as well as the average and peak transmit power constraints. For the U2G case, the formulated joint trajectory optimization and power control problem for average secrecy rate maximization is difficult to be solved directly due to its non-smooth objective function. To tackle this difficulty, we reformulate the problem into an equivalent problem with a smooth objective function without loss of optimality. Although the non-smoothness issue is resolved, the reformulated problem is still non-convex due to the coupling of the transmit power and UAV trajectory optimization variables. We thus propose an efficient iterative algorithm for solving this problem approximately based on the block coordinate descent method. Specifically, we divide the optimization variables into two blocks, one for transmit power control and the other for UAV trajectory optimization. Then the two blocks of variables are optimized alternately in an iterative manner, i.e., in each iteration one block is optimized with the other block fixed and vice versa. One corresponding sub-problem that optimizes the UAV trajectory under given transmit power is still difficult to solve due to its non-convexity. We thus apply the successive convex optimization method to solve the problem approximately. Finally, we show that our proposed joint optimization algorithm is guaranteed to converge. On the other hand, for the G2U case, similar to the U2G case, we also propose an efficient algorithm to solve the formulated problem by using the block coordinate descent and successive convex optimization methods, while some modification is made in the problem formulation to deal with the non-LoS channel of the eavesdropper link in this case. Simulation results show that the proposed joint trajectory and transmit power designs can improve the average secrecy rates in both U2G and G2U communications, as compared to other benchmark schemes without applying the trajectory optimization and/or transmit power control. Furthermore, it is observed that trajectory optimization and transmit power control are both essential for the U2G case, while for the G2U case, trajectory optimization is less effective as compared to power control.

It is worth noting that UAV systems, there have been prior works (e.g., \cite{He2017, He2018, Singandhupe2018}) that address the security and safety issues of UAVs from other perspectives. In \cite{He2017}, the security vulnerabilities in the global positioning system (GPS) spoofing attack and WiFi attack in UAV applications have been addressed, and effective solutions to these attacks have been suggested. In \cite{He2018}, a monocular camera and inertial measurement unit (IMU) sensor based GPS spoofing detection scheme and an image localization approach for UAV autonomous return have been proposed to support the security and safety of UAVs. In \cite{Singandhupe2018}, a biometric system based on encryption has been proposed to secure the communication link between a UAV and a BS on the ground. Note that these prior works are fundamentally different from this paper which applies the physical layer security technique to deal with the eavesdropping attack in UAV-ground communication systems. It is also noted that there have been prior works (e.g., \cite{WuGC20170,WuGC2017,WuGC20172,WuGC20173,Jiang2012, Zeng2016Trans, Zeng2017}) on trajectory optimization for various UAV communication systems, which consider different system setups and design objectives. In \cite{WuGC20170}, three fundamental tradeoffs in UAV-enabled wireless networks have been identified, i.e., throughput-delay tradeoff, throughput-energy tradeoff, and delay-energy tradeoff.  In \cite{WuGC2017}, a UAV mobile BS serving multiple users is considered, where the UAV trajectory and multiuser scheduling are jointly designed to maximize the minimum throughput of the users. In particular, it is shown in \cite{WuGC20172} that significant communication throughput gains can be achieved by mobile UAVs over static UAVs/fixed terrestrial BSs by exploiting the new design degree of freedom of UAV trajectory optimization, especially for delay-tolerant applications.  To study the fundamental limits of the UAV-enabled wireless network, the capacity region of a two-user broadcast channel is characterized in \cite{WuGC20173} where it has been rigorously  proved that a simple ``fly-hover-fly'' trajectory is capacity achieving. In \cite{Zeng2016Trans}, a UAV-enabled mobile relaying system is investigated, where the UAV trajectory and transmit power are jointly designed to maximize the throughput. In \cite{Jiang2012}, the UAV flying heading is optimized to maximize the achievable sum rate from ground nodes to a UAV by assuming a constant flying speed. In \cite{Zeng2017}, a new design paradigm that jointly considers both the communication throughput and the UAV's flying energy consumption is proposed to maximize the energy efficiency of a point-to-point U2G communication system. Different from these prior works, in this paper, we apply both trajectory optimization and transmit power control to maximize the secrecy rates of both U2G and G2U communications. The main contributions of this paper are highlighted as follows.
\begin{itemize}
	\item Compared to the existing physical layer security literature, this paper is the first to exploit the high mobility of UAVs to improve the secrecy rate via joint trajectory and power control optimization.
	\item Both the U2G and G2U cases in UAV-ground communications are considered. The considered problems for both the two cases are difficult to be solved optimally due to their non-smooth and non-concave objective functions. To tackle this difficulty, we first reformulate the problems into equivalent problems with smooth objective functions, and then propose efficient algorithms to solve the reformulated problems approximately based on the block coordinate descent method and the successive convex optimization method. The obtained results show the fundamental secrecy rate limits of the U2G and G2U communications and demonstrate the importance and necessity of the joint UAV trajectory and transmit power optimization in maximizing the secrecy rate for the new settings. Moreover, the obtained results provide different design guidelines for the U2G case and the G2U case, respectively.
\end{itemize}

The remainder of this paper is organized as follows. Section II presents the system model and problem formulation. Sections III and IV present joint trajectory optimization and transmit power control algorithms for the U2G and G2U cases, respectively. Section V provides simulation results to validate the performance of the proposed algorithms as compared to three benchmark schemes. Finally, Section VI concludes the paper.

\section{System Model and Problem Formulation}
\subsection{System Model}
As shown in Fig. \ref{FigDLULModel}, we consider a UAV-enabled wireless communication system where a UAV above ground and a node on the ground communicate with each other, while a potential eavesdropper on the ground aims to intercept the communications between them. Without loss of generality, we consider a three-dimensional (3D) Cartesian coordinate system with the ground node and the eavesdropper located at $(x_{\text{G}}, y_{\text{G}}, 0)$ and $(x_{\text{E}}, y_{\text{E}}, 0)$ in meters (m), respectively. Their locations are assumed to be fixed and known to the UAV, where the location of the eavesdropper can be detected by using an optical camera or SAR equipped on the UAV. On the other hand, the obtained secrecy rate when the location of the eavesdropper is known serves as an upper bound for that when the location of the eavesdropper is not known.

We consider a given finite flight period of the UAV, with the duration denoted by $T$ in seconds (s). It is assumed that the UAV flies at a fixed altitude of $H$ in m above ground, which can be considered as the minimum altitude required for safety considerations such as terrain or building avoidance. The coordinate of the UAV over time is denoted as $( x(t),y(t),H )$ in m, $0 \leq t \leq T$. For convenience, we divide the period $T$ into $N$ time slots with equal length, i.e., $T=N d_t$, with $d_t$ in s denoting the length of a time slot, which is chosen sufficiently small such that the UAV's location can be regarded as unchanged within each time slot from the viewpoint of the ground node. As a result, the UAV's coordinate in slot $n$ can be denoted as $( x[n],y[n],H )$, and the UAV's horizontal trajectory $(x(t),y(t))$ over the flight period $T$ can be approximated by the sequence $\{ x[n],y[n] \}_{n=1}^{N}$. Denote the maximum speed of the UAV as $v_{\max} $ in m/s. Thus, the maximum flying distance of the UAV in each slot is $D=v_{\max} d_t$. The initial and final locations of the UAV are assumed to be given, which are denoted by $(x_0,y_0,H)$ and $(x_F,y_F,H)$ in m, respectively. For the UAV trajectory to be feasible, we assume that the distance between the initial and final location satisfies that $\sqrt{(x_F - x_0)^2+(y_F - y_0)^2} \leq v_{\max} T$. As a result, the mobility constraints of the UAV can be expressed as
\begin{subequations}  \label{EquMobilityCon}
\begin{align}
( x[1] - x_0)^2 + (y[1] - y_0)^2 & \leq D^2,  \label{EquInitialCon}   \\
( x[n+1]-x[n] )^2 + ( y[n+1]-y[n] )^2 & \leq D^2, \nonumber \\
n=1,\ldots,N-1, \quad \quad \quad \quad \quad \quad \quad \quad \; \; &     \label{EquMidCon}  \\
( x_F - x[N] )^2 + (y_F - y[N])^2 & \leq D^2.   \label{EquFinalCon}
\end{align}
\end{subequations}

We consider both the U2G and G2U communications in the system of interest, which are specified in detail in the following, respectively.

\subsubsection{U2G Transmission}
In the U2G case, the UAV and the ground node play the roles of legitimate transmitter and receiver, respectively. The legitimate link from the UAV to the ground node and the eavesdropping link from the UAV to the eavesdropper are both assumed to be LoS channels, as the recent measurement results in \cite{Lin2017} have shown that the LoS model offers a good approximation for practical UAV-ground communications. Thus, the LoS channel power gain from the UAV to the ground node in time slot $n$ follows the free-space path loss model, given by
\begin{equation}
g_{\text{UG}}[n] = \beta_0 d_{\text{UG}}^{-2}[n] = \frac{ \beta_0 }{ ( x[n] - x_{\text{G}} )^2 + ( y[n] - y_{\text{G}} )^2 + H^2 },
\end{equation}
where $\beta_0$ denotes the channel power gain at the reference distance $d_0=1$m, which depends on the carrier frequency and the antenna gains of the transmitter and receiver, and $d_{\text{UG}}[n] = \sqrt{ ( x[n] - x_{\text{G}} )^2 + ( y[n] - y_{\text{G}} )^2 + H^2 }$ is the distance from the UAV to the ground node in time slot $n$. Similarly, the LoS channel power gain from the UAV to the eavesdropper in time slot $n$ is given by
\begin{equation}
g_{ \text{UE} } [n] = \frac{ \beta_0 }{ ( x[n] - x_{\text{E}} )^2 + ( y[n] - y_{\text{E}} )^2 + H^2 }.
\end{equation}

We denote $p[n]$ as the transmit power of the UAV in time slot $n$. In practice, $ p[n] $'s are usually subject to both average and peak limits over time, denoted by $\bar{P}$ and $P_{\text{peak}}$, respectively. Thus, the transmit power constraints are expressed as
\begin{subequations} \label{EquPowerCon}
\begin{align}
\frac{1}{N} \sum_{n=1}^{N} p[n] & \leq \bar{P},  \label{EquAvgPowCon}  \\
0 \leq p[n] & \leq P_{\text{peak}}, \; \forall n.   \label{EquPeakPowCon}
\end{align}
\end{subequations}
To make the constraint in \eqref{EquAvgPowCon} non-trivial, we assume $\bar{P}<P_{\text{peak}}$ in this paper. In the absence of the eavesdropper, the achievable rate from the UAV to the ground node in bits/second/Hertz (bps/Hz) in time slot $n$ can be expressed as
\begin{align}
R_{\text{UG}}[n] = & \log_2 \left(1 + \frac{ p[n] g_{\text{UG}}[n] }{ \sigma^2 }  \right)     \nonumber \\
= &   \log_2 \left( 1+ \frac{ \gamma_0 p[n] } {  ( x[n] - x_{\text{G}} )^2 + ( y[n] - y_{\text{G}} )^2 + H^2 }  \right),  \label{EquRAB}
\end{align}
where $\sigma^2$ is the additive white Gaussian noise (AWGN) power at the receiver and $\gamma_0 =\beta_0 / \sigma^2$ is the reference signal-to-noise ratio (SNR). Similarly, the achievable rate from the UAV to the eavesdropper in bps/Hz in time slot $n$ is given by
\begin{equation}  \label{EquRAE}
R_{\text{UE}}[n] = \log_2 \left( 1+ \frac{ \gamma_0 p[n] }{ ( x[n] - x_{\text{E}} )^2 + ( y[n] - y_{\text{E}} )^2 + H^2 } \right).
\end{equation}
With \eqref{EquRAB} and \eqref{EquRAE}, the average secrecy rate achievable for the U2G link in bps/Hz over the total $N$ time slots is given by \cite{Gopala2008}
\begin{align}
& R_{\text{sec}}^{(\text{U2G})}  \nonumber \\
= & \frac{1}{N} \sum_{n=1}^N \bigg[ \log_2 \left( 1+ \frac{ \gamma_0 p[n] } {  ( x[n] - x_{\text{G}} )^2 + ( y[n] - y_{\text{G}} )^2 + H^2 }  \right) \nonumber  \\
&  - \log_2 \left( 1+ \frac{ \gamma_0 p[n] }{ ( x[n] - x_{\text{E}} )^2 + ( y[n] - y_{\text{E}} )^2 + H^2 } \right) \bigg]^{+},  \label{EquSecrecyRate}
\end{align}
where $[x]^+ \triangleq \max(x,0)$.

\subsubsection{G2U Transmission}
In the G2U case, the ground node and the UAV play the roles of legitimate transmitter and receiver, respectively. The legitimate channel from the ground node to the UAV is assumed to be LoS, similar as in the U2G case, whose channel power gain in time slot $n$ is given by
\begin{equation}
g_{ \text{GU} } [n] = \frac{ \beta_0 }{ ( x[n] - x_{\text{G}} )^2 + ( y[n] - y_{\text{G}} )^2 + H^2 }.
\end{equation}
Since both the ground node and the eavesdropper are on the ground, the eavesdropping channel between them is assumed to constitute both distance-dependent path loss with pass-loss exponent $\kappa \geq 2$ and small-scale Rayleigh fading. Thus, the channel power gain from the ground node to the eavesdropper at any time is given by
\begin{equation}   \label{EquChGain}
g_{ \text{GE} } = \frac{ \beta_0  }{ d_{\text{GE}}^{\kappa} }  \zeta,
\end{equation}
where $d_{\text{GE}} = \sqrt{ (x_{\text{G}} - x_{\text{E}})^2 + (y_{\text{G}} - y_{\text{E}})^2 }$ denotes the distance between the ground node and the eavesdropper, and $\zeta$ is an exponentially distributed random variable with unit mean accounting for the Rayleigh fading.

We denote $q[n]$ as the transmit power of the ground node in time slot $n$. Similar to the U2G case, $q[n]$'s are constrained by average power limit $\bar{Q}$ and peak power limit $Q_{\text{peak}}$, i.e.,
\begin{subequations} \label{EquPowerConUL}
\begin{align}
\frac{1}{N} \sum_{n=1}^{N} q[n] & \leq \bar{Q},  \label{EquAvgPowConUL}  \\
0 \leq q[n] & \leq Q_{\text{peak}}, \; \forall n,   \label{EquPeakPowConUL}
\end{align}
\end{subequations}
where $\bar{Q} < Q_{\text{peak}}$ is assumed. Similar to \eqref{EquRAB}, the achievable rate from the ground node to the UAV in bps/Hz in time slot $n$ can be expressed as
\begin{equation}  \label{EquRABUL}
R_{\text{GU}}[n] =  \log_2 \left( 1+ \frac{ \gamma_0 q[n] } {  ( x[n] - x_{\text{G}} )^2 + ( y[n] - y_{\text{G}} )^2 + H^2 }  \right).
\end{equation}
The achievable rate from the ground node to the eavesdropper in bps/Hz in time slot $n$ is expressed as
\begin{subequations}  \label{EquRAEUL}
\begin{align}
R_{\text{GE}}[n] = & \; \mathbb{E}_{\zeta} \left[ \log_2 \left( 1 + \frac{\gamma_0 q[n]}{ d_{\text{GE}}^{\kappa} } \zeta  \right) \right]  \label{EquRAEUL1}  \\
\leq & \; \log_2 \left( 1 + \frac{\gamma_0 q[n]}{ d_{\text{GE}}^{\kappa} } \mathbb{E}_{\zeta} \left[ \zeta \right]  \right) \label{EquRAEUL2} \\
=& \;  \log_2 \left( 1 + \frac{\gamma_0 q[n]}{ d_{\text{GE}}^{\kappa} }   \right),  \label{EquRAEUL3}
\end{align}
\end{subequations}
where $\mathbb{E}_{\zeta} [ \cdot ]$ in \eqref{EquRAEUL1} denotes the mathematical expectation with respect to random variable $\zeta$, and the inequality in \eqref{EquRAEUL2} is due to Jensen's inequality and the fact that $\log_2 ( 1 + \gamma_0 q[n] \zeta /d_{\text{GE}}^{\kappa})$ is concave with respect to $\zeta$. \eqref{EquRAEUL3} shows an upper bound of $R_{\text{GE}}[n]$. We consider the worst-case secrecy rate performance by assuming that the eavesdropper is able to achieve this upper bound. With \eqref{EquRABUL} and \eqref{EquRAEUL3}, the following average secrecy rate of the G2U link in bps/Hz over the total $N$ time slots is thus achievable,
\begin{align}
& R_{\text{sec}}^{(\text{G2U})}  \nonumber \\
= & \frac{1}{N} \sum_{n=1}^N  \bigg[ \log_2 \left( 1+ \frac{ \gamma_0 q[n] } {  ( x[n] - x_{\text{G}} )^2 + ( y[n] - y_{\text{G}} )^2 + H^2 }  \right)  \nonumber \\
& -  \log_2 \left( 1 + \frac{\gamma_0 q[n]}{ d_{\text{GE}}^{\kappa} }   \right)   \bigg]^{+}.  \label{EquSecrecyRateUL}
\end{align}

\subsection{Problem Formulation}
For the U2G case, our goal is to maximize the average secrecy rate $R_{\text{sec}}^{(\text{U2G})}$ in \eqref{EquSecrecyRate} by jointly optimizing the UAV's transmit power $\mathbf{p} \triangleq \left[p[1],\ldots,p[N] \right]^{\dagger}$ and the UAV's trajectory in terms of its horizontal coordinates $\mathbf{x} \triangleq \left[x[1], \ldots, x[N]\right]^{\dagger}$ and $\mathbf{y} \triangleq \left[y[1], \ldots, y[N]\right]^{\dagger}$ over all the $N$ time slots, where the superscript $\dagger$ denotes the transpose operation. The optimization variables are subject to the UAV's mobility constraints in \eqref{EquMobilityCon} and the average and peak transmit power constraints in \eqref{EquPowerCon}. We formulate the secrecy rate maximization problem as follows (by dropping the constant term $1/N$ in \eqref{EquSecrecyRate})\footnote{Generally, the UAV's flying altitude can also be optimized by adding the minimum and the maximum altitude constraints. However, it is easy to verify that for our considered problem the optimal objective value can be always achieved at the minimum UAV altitude under the LoS air-to-ground channel model.}
\begin{align}
\text{(P1)}:  &  \nonumber \\
\; \max_{ \mathbf{x}, \mathbf{y}, \mathbf{p} } & \; \sum_{n=1}^N \bigg[ \log_2 \left( 1+ \frac{ \gamma_0 p[n] } {  ( x[n] - x_{\text{G}} )^2 + ( y[n] - y_{\text{G}} )^2 + H^2 }  \right) \nonumber \\
& \; - \log_2 \left(1+ \frac{ \gamma_0 p[n] } { ( x[n] - x_{\text{E}} )^2 + ( y[n] - y_{\text{E}} )^2 + H^2 }  \right) \bigg]^+  \label{EquOriginal} \\
\text{s.t.} &   \;  \eqref{EquMobilityCon}, \; \eqref{EquPowerCon}. \nonumber
\end{align}

Similarly, for the G2U case, we maximize $R_{\text{sec}}^{(\text{G2U})}$ in \eqref{EquSecrecyRateUL} by jointly optimizing the ground node's transmit power $\mathbf{q} \triangleq \left[q[1],\ldots,q[N] \right]^{\dagger}$ and the UAV's horizontal trajectory $\mathbf{x}$ and $\mathbf{y}$. The problem is thus formulated as
\begin{align}
\text{(P2)}:  &  \nonumber  \\
\max_{ \mathbf{x}, \mathbf{y}, \mathbf{q} } & \; \sum_{n=1}^N  \bigg[ \log_2 \left( 1+ \frac{ \gamma_0 q[n] } {  ( x[n] - x_{\text{G}} )^2 + ( y[n] - y_{\text{G}} )^2 + H^2 }  \right) \nonumber  \\
& \; - \log_2 \left( 1 + \frac{\gamma_0 q[n]}{ d_{\text{GE}}^{\kappa} }   \right)   \bigg]^{+}  \label{EquOriginalUL} \\
\text{s.t.} & \;  \eqref{EquMobilityCon}, \; \eqref{EquPowerConUL}. \nonumber
\end{align}
Note that different from problem (P1), only the first logarithmic function in the objective of (P2), i.e., $\log_2 \big( 1+ \frac{ \gamma_0 q[n] } {  ( x[n] - x_{\text{G}} )^2 + ( y[n] - y_{\text{G}} )^2 + H^2 }  \big) $, contains the UAV trajectory variables. This is because the achievable rate from the ground node to the eavesdropper does not depend on the trajectory of the UAV.

Problems (P1) and (P2) are both difficult to be solved optimally due to the following two reasons. First, the operator $[\cdot]^+$ makes the objective functions of (P1) and (P2) non-smooth at zero value. Second, even without $[\cdot]^+$, their objective functions are non-concave with respect to either $\mathbf{x}$, $\mathbf{y}$, or $\mathbf{p}$. In Sections III and IV, we propose efficient algorithms for solving problems (P1) and (P2) approximately, respectively.

\section{Proposed Algorithm for Problem (P1)}
First, we consider problem (P1) for the U2G case. To handle the non-smoothness of the objective function of (P1), the following lemma is used.

\begin{lemma}
Problem (P1) has the same optimal value as that of the following problem,
\begin{align}
\text{(P3)}:  & \nonumber \\
 \max_{ \mathbf{x}, \mathbf{y}, \mathbf{p} } & \; \sum_{n=1}^N \bigg[ \log_2 \left( 1+ \frac{ \gamma_0 p[n] } {  ( x[n] - x_{\text{G}} )^2 + ( y[n] - y_{\text{G}} )^2 + H^2 }  \right) \nonumber \\
& \; - \log_2 \left(1+ \frac{ \gamma_0 p[n] } { ( x[n] - x_{\text{E}} )^2 + ( y[n] - y_{\text{E}} )^2 + H^2 }  \right) \bigg]  \\
\text{s.t.} & \;  \eqref{EquMobilityCon}, \; \eqref{EquPowerCon}. \nonumber
\end{align}
\end{lemma}

\begin{proof}
Denote $L_1$ and $L_3$ as the optimal values of (P1) and (P3), respectively. First, we have $L_1 \geq L_3$, since the objective function of (P1) is no smaller than that of (P3), and (P1) and (P3) have the same constraints.

Next, we show $L_3 \geq L_1$ also holds. Denote $(\mathbf{x}^*, \mathbf{y}^*, \mathbf{p}^*)$ as the optimal solution to (P1), where $\mathbf{x}^* = [x^*[1],\ldots,x^*[N]]^\dagger$, $\mathbf{y}^*=[y^*[1],\ldots,y^*[N]]^\dagger$, and $\mathbf{p}^*=[p^*[1],\ldots,p^*[N]]^\dagger$. Define
\begin{align}
& f(x[n],y[n],p[n]) \nonumber \\
\triangleq & \log_2 \left( 1+ \frac{ \gamma_0 p[n] } {  ( x[n] - x_{\text{G}} )^2 + ( y[n] - y_{\text{G}} )^2 + H^2 }  \right)  \nonumber  \\
& - \log_2 \left(1+ \frac{ \gamma_0 p[n] } { ( x[n] - x_{\text{E}} )^2 + ( y[n] - y_{\text{E}} )^2 + H^2 }  \right).  \nonumber
\end{align}
We construct a feasible solution to (P3), termed $(\tilde{\mathbf{x}}, \tilde{\mathbf{y}}, \tilde{\mathbf{p}})$, such that $\tilde{\mathbf{x}} = \mathbf{x}^*$, $\tilde{\mathbf{y}} = \mathbf{y}^*$, and the elements of $\tilde{\mathbf{p}}$ are obtained as
\begin{displaymath}
\tilde{p}[n] = \begin{cases} p^*[n] & f(x^*[n],y^*[n],p^*[n]) \geq 0, \\ 0 & f(x^*[n],y^*[n],p^*[n])< 0.    \end{cases}
\end{displaymath}
Denote $\tilde{L}$ as the objective value of (P3) attained at $(\tilde{\mathbf{x}}, \tilde{\mathbf{y}}, \tilde{\mathbf{p}})$. The newly constructed solution $(\tilde{\mathbf{x}}, \tilde{\mathbf{y}}, \tilde{\mathbf{p}})$ ensures that $\tilde{L} = L_1$. Since $(\tilde{\mathbf{x}}, \tilde{\mathbf{y}}, \tilde{\mathbf{p}})$ is a feasible solution to (P3), it follows that $L_3 \geq \tilde{L}$, and thus $L_3 \geq L_1$. Therefore, $L_1 = L_3$, which completes the proof.
\end{proof}

Based on Lemma 1, we only need to focus on solving problem (P3). Although problem (P3) resolves the non-smoothness issue, it is still non-convex and difficult to solve. However, we observe that the constraint \eqref{EquMobilityCon} contains only the variables $(\mathbf{x}, \mathbf{y})$ for UAV trajectory optimization and the constraint \eqref{EquPowerCon} contains only the variables $\mathbf{p}$ for transmit power control. As such, the optimization variables of (P3) can be partitioned into two blocks, i.e., $\mathbf{p}$ and $(\mathbf{x}, \mathbf{y})$, respectively, which facilitates the development of an iterative algorithm for solving problem (P3) by applying the block coordinate descent method. Specifically, we solve problem (P3) by solving the following two sub-problems iteratively: one (denoted by sub-problem 1) optimizes the transmit power $\mathbf{p}$ under given UAV trajectory $(\mathbf{x}, \mathbf{y})$, while the other (denoted by sub-problem 2) optimizes the UAV trajectory $(\mathbf{x}, \mathbf{y})$ under given transmit power $\mathbf{p}$, as detailed in the next two subsections, respectively. Then, we present the overall algorithm and show that it is guaranteed to converge.

\subsection{Sub-Problem 1: Optimizing Transmit Power Given UAV Trajectory}  \label{SecPCDL}
For given UAV trajectory $(\mathbf{x}, \mathbf{y})$, sub-problem 1 can be expressed as
\begin{align}
\max_{\mathbf{p}} \; &  \sum_{n =1}^N \big[ \log_2 \left( 1+ a_n p[n]   \right)  - \log_2 \left(1+ b_n p[n] \right)  \big]   \label{EquSubProb1} \\
\text{s.t.} \; \; & \eqref{EquPowerCon} , \nonumber
\end{align}
where
\begin{equation}  \label{Equa}
a_n = \frac{ \gamma_0 } {  ( x[n] - x_{\text{G}} )^2 + ( y[n] - y_{\text{G}} )^2 + H^2 },
\end{equation}
\begin{equation}   \label{Equb}
b_n = \frac{ \gamma_0  } { ( x[n] - x_{\text{E}} )^2 + ( y[n] - y_{\text{E}} )^2 + H^2 }.
\end{equation}
Although problem \eqref{EquSubProb1} is non-convex, it has been shown in \cite{Gopala2008} and \cite{ZhangTVT2018} that the optimal solution can be obtained as
\begin{equation}   \label{EquOptPowerScheme}
p^{*}[n] = \begin{cases}
\min \left( [ \hat{p}[n] ]^+ , P_{\text{peak}} \right)  & a_n > b_n , \\
0  & a_n \leq b_n,
\end{cases}
\end{equation}
where
\begin{equation}  \label{EquOptPowSol}
\hat{p}[n] = \sqrt{ \left( \frac{1}{2b_n} - \frac{1}{2a_n}  \right)^2 + \frac{1}{\lambda \ln2} \left( \frac{1}{b_n} - \frac{1}{a_n}  \right)  }  -\frac{1}{2b_n} -  \frac{1}{2a_n}.
\end{equation}
In \eqref{EquOptPowSol}, $\lambda \geq 0$ is a constant that ensures the average power constraint $\frac{1}{N}\sum_{ n =1 }^N p[n] \leq \bar{P}$ to be satisfied when the optimal solution of problem \eqref{EquSubProb1} is attained, which can be found efficiently via a one-dimensional bisection search \cite{ZhangTVT2018, Boyd2004}.

\subsection{Sub-Problem 2: Optimizing UAV Trajectory Given Transmit Power}  \label{SecSubProb2}
For given transmit power $\mathbf{p}$, by letting $P_n =  \gamma_0 p[n]$, we can express sub-problem 2 as
\begin{align}
\max_{ \mathbf{x}, \mathbf{y} } \; &  \sum_{n=1}^N \bigg[  \log_2 \left( 1+ \frac{ P_n } {  ( x[n] - x_{\text{G}} )^2 + ( y[n] - y_{\text{G}} )^2 + H^2 }  \right)  \nonumber  \\
& -  \log_2 \left(1+ \frac{ P_n } { ( x[n] - x_{\text{E}} )^2 + ( y[n] - y_{\text{E}} )^2 + H^2 } \right)  \bigg]  \label{EquSubProb2Ori}  \\
\text{s.t.}  \; \; & \eqref{EquMobilityCon}. \nonumber
\end{align}
Note that the objective function of problem \eqref{EquSubProb2Ori} is non-concave with respect to $\mathbf{x}$ and $\mathbf{y}$, so it is a non-convex optimization problem and cannot be solved optimally in general. By introducing slack variables $\mathbf{t} \triangleq \left[ t[1], \ldots, t[N]  \right]^{\dagger}$ and $\mathbf{u} \triangleq \left[ u[1], \ldots, u[N]  \right]^{\dagger}$, we first consider the following problem,
\begin{subequations}   \label{EquSubProb2Reform}
\begin{align}
\max_{ \mathbf{x}, \mathbf{y}, \mathbf{t}, \mathbf{u} } \; &  \sum_{n=1}^N \bigg[  \log_2 \left( 1+ \frac{ P_n } { u[n]  }  \right)  -  \log_2 \left(1+ \frac{ P_n } { t[n] } \right)   \bigg]\label{EquObjSub2}    \\
\text{s.t.}  \; \; \;  &  t[n] - x^2[n] + 2x_{\text{E}}x[n] - x_{\text{E}}^2  - y^2[n] + 2y_{\text{E}}y[n] \nonumber \\
& - y_{\text{E}}^2 - H^2 \leq 0 , \; \forall n,  \label{EquCont}  \\
&  x^2[n] - 2x_{\text{G}}x[n] + x_{\text{G}}^2 +  y^2[n]  - 2y_{\text{G}}y[n] + y_{\text{G}}^2 \nonumber \\
& + H^2 -u[n] \leq 0, \;  \forall n,  \label{EquConu} \\
&  t[n] \geq H^2, \; \forall n , \label{EquContGeq0} \\
& \eqref{EquMobilityCon}.  \nonumber
\end{align}
\end{subequations}
At the optimal solution of problem \eqref{EquSubProb2Reform}, constraints \eqref{EquCont} and \eqref{EquConu} should hold with equalities, since otherwise $t[n]$ ($u[n]$) can be increased (decreased) to improve the objective value. Therefore, problems \eqref{EquSubProb2Ori} and \eqref{EquSubProb2Reform} have the same optimal value and optimal solution of $(\mathbf{x},\mathbf{y})$. Next, we focus on solving problem \eqref{EquSubProb2Reform}.

The term $ \log_2 \left( 1+ \frac{ P_n } { u[n]  }  \right)$ in \eqref{EquObjSub2} is convex with respect to $u[n]$, and the terms $-x^2[n]$ and $-y^2[n]$ in \eqref{EquCont} are concave with respect to $x[n]$ and $y[n]$, respectively. However, a maximization problem with a non-concave objective function and/or a non-convex feasible region is in general non-convex and thus difficult to be solved optimally. Based on the facts that the first-order Taylor expansion of a convex function is its global under-estimator and that of a concave function is its global over-estimator \cite{Boyd2004}, we propose an iterative algorithm to solve problem \eqref{EquSubProb2Reform} approximately by applying the successive convex optimization method. The algorithm obtains an approximate solution to problem \eqref{EquSubProb2Reform} by maximizing a concave lower bound of its objective function within a convex feasible region, which is detailed as follows. First, the algorithm assumes a given initial point $( \mathbf{x}_{\text{fea}}, \mathbf{y}_{\text{fea}}, \mathbf{u}_{\text{fea}})$ which is feasible to \eqref{EquSubProb2Reform}, where $\mathbf{x}_{\text{fea}} \triangleq \left[ x_{\text{fea}}[1], \ldots, x_{\text{fea}}[N] \right]^{\dagger}$, $\mathbf{y}_{\text{fea}} \triangleq \left[ y_{\text{fea}}[1], \ldots, y_{\text{fea}}[N] \right]^{\dagger}$, and $\mathbf{u}_{\text{fea}} \triangleq \left[ u_{\text{fea}}[1], \ldots, u_{\text{fea}}[N] \right]^{\dagger}$. Then, by using the first-order Taylor expansions of $ \log_2 \left( 1+ \frac{ P_n } { u[n]  }  \right)$, $-x^2[n]$, and $-y^2[n]$ at the points $u_{\text{fea}}[n]$, $x_{\text{fea}}[n]$, and $y_{\text{fea}}[n]$, respectively, i.e.,
\begin{align}
 \log_2 \left( 1+ \frac{ P_n } { u[n]  }  \right) \geq & \log_2 \left( 1+ \frac{ P_n } { u_{\text{fea}}[n]  }  \right) \nonumber \\
 &  - \frac{ P_n ( u[n] - u_{\text{fea}}[n] ) }{ \ln2 (u_{\text{fea}}^2[n] + P_n u_{\text{fea}}[n] ) },  \label{EquTayloru}
\end{align}
\begin{equation}   \label{EquTaylorx}
-x^2[n] \leq x_{\text{fea}}^2[n] - 2 x_{\text{fea}}[n] x[n],
\end{equation}
\begin{equation}   \label{EquTaylory}
-y^2[n] \leq y_{\text{fea}}^2[n] - 2 y_{\text{fea}}[n] y[n],
\end{equation}
problem \eqref{EquSubProb2Reform} is approximated as
\begin{subequations}  \label{EquSubProb2Traject}
\begin{align}
\max_{ \mathbf{x},\mathbf{y}, \mathbf{t},\mathbf{u} } \; & \sum_{n=1}^N \bigg[ - \frac{ P_n u[n] }{ \ln 2 (  u_{\text{fea}}^2[n] + P_n u_{\text{fea}}[n] ) }  -  \log_2 \left(1+ \frac{ P_n } { t[n] } \right)    \bigg]\\
\text{s.t.} \; \;  &  t[n] + x_{\text{fea}}^2[n]  - 2 x_{\text{fea}}[n] x[n] + 2x_{\text{E}}x[n] - x_{\text{E}}^2 + y_{\text{fea}}^2[n] \nonumber \\
&  - 2 y_{\text{fea}}[n] y[n]  + 2y_{\text{E}}y[n] - y_{\text{E}}^2 - H^2 \leq 0  , \; \forall n,  \label{EquXYConUB} \\
& \eqref{EquMobilityCon}, \; \eqref{EquConu},\; \eqref{EquContGeq0}. \nonumber
\end{align}
\end{subequations}
After such approximation, we note that the objective function of problem \eqref{EquSubProb2Traject} is concave and its feasible region is convex. Thus, problem \eqref{EquSubProb2Traject} is a convex optimization problem, which can be optimally solved by the interior-point method \cite{Boyd2004}. Since the first-order Taylor expansions of $-x^2[n]$ and $-y^2[n]$ are their global over-estimators, any solution $(x[n], y[n])$ satisfying \eqref{EquXYConUB} will satisfy \eqref{EquCont}. As a result, the solution of problem \eqref{EquSubProb2Traject} is guaranteed to be a feasible solution of problem \eqref{EquSubProb2Reform}. Moreover, the first-order Taylor expansion of $ \log_2 \left( 1+ \frac{ P_n } { u[n]  }  \right)$ is its global under-estimator. As such, problem \eqref{EquSubProb2Traject} maximizes a lower bound of the objective function of problem \eqref{EquSubProb2Reform}, and the lower bound and the objective function of \eqref{EquSubProb2Reform} are equal only at the given point $(\mathbf{x}_{\text{fea}}, \mathbf{y}_{\text{fea}}, \mathbf{u}_{\text{fea}})$; thus, the objective value of problem \eqref{EquSubProb2Reform} with the solution obtained by solving problem \eqref{EquSubProb2Traject} is no smaller than that with the given point $(\mathbf{x}_{\text{fea}}, \mathbf{y}_{\text{fea}}, \mathbf{u}_{\text{fea}})$.

\subsection{Overall Algorithm}  \label{SecOverallAlg}

\begin{algorithm}
	\caption{Proposed Algorithm for Problem (P1).}
	\begin{algorithmic}[1]   \label{AlgDownlink}
		\STATE \textbf{Initialization:} Set $k=0$. Find an initial feasible solution $( \mathbf{p}^{(0)} , \mathbf{x}^{(0)}, \mathbf{y}^{(0)} )$ and an initial slack variable $\mathbf{u}^{(0)}$. Set $R^{(0)} = f_{\text{(P3)}} ( \mathbf{p}^{(0)} , \mathbf{x}^{(0)}, \mathbf{y}^{(0)}  )$.
		\REPEAT
		\STATE Set $k = k+1$.
		\STATE With given $\mathbf{p}^{(k-1)}$, set the feasible solution $\mathbf{x}_{\text{fea}} = \mathbf{x}^{(k-1)}$, $\mathbf{y}_{\text{fea}} = \mathbf{y}^{(k-1)}$ and $\mathbf{u}_{\text{fea}} = \mathbf{u}^{(k-1)}$, then update the trajectory variable $(\mathbf{x}^{(k)}, \mathbf{y}^{(k)})$ and the slack variable $\mathbf{u}^{(k)}$ by solving problem \eqref{EquSubProb2Traject}.
        \STATE With given $(\mathbf{x}^{(k)}, \mathbf{y}^{(k)} )$, update the transmit power control variable $\mathbf{p}^{(k)}$ using \eqref{EquOptPowerScheme}.
        \STATE Set $R^{(k)} = f_{\text{(P3)}} ( \mathbf{p}^{(k)} , \mathbf{x}^{(k)}, \mathbf{y}^{(k)}  )$.
		\UNTIL {$\frac{R^{(k)} - R^{(k-1)}}{R^{(k)}} < \epsilon$.}
	\end{algorithmic}
\end{algorithm}

In summary, the overall algorithm can find a suboptimal solution to problem (P1) by applying the block coordinate descent method, and solves the two sub-problems \eqref{EquSubProb1} and \eqref{EquSubProb2Traject} alternately in an iterative manner. The details of the proposed algorithm are summarized in Algorithm \ref{AlgDownlink}, where $R^{(k)} = f_{\text{(P3)}} ( \mathbf{p}^{(k)}, \mathbf{x}^{(k)}, \mathbf{y}^{(k)}  )$ denotes the objective value of problem (P3) with variables $\mathbf{p}^{(k)}$, $\mathbf{x}^{(k)}$, and $\mathbf{y}^{(k)}$ in iteration $k$, and $\epsilon$ denotes a small positive threshold indicating the accuracy of convergence. The convergence of Algorithm \ref{AlgDownlink} is proved as follows. First, we show that in iteration $k$ ($k \geq 1$) of Algorithm \ref{AlgDownlink}, the objective value of problem (P3) is non-decreasing after executing steps 4 and 5. Denote $\phi(\mathbf{x}, \mathbf{y}, \mathbf{p})$ as the objective value of problem (P3), and $\xi(\mathbf{x}, \mathbf{y}, \mathbf{u}, \mathbf{p})$ and $\xi_{lb}(\mathbf{x}, \mathbf{y}, \mathbf{u}, \mathbf{p})$ as the objective values of problems \eqref{EquSubProb2Reform} and \eqref{EquSubProb2Traject}, respectively. In step 4, we have the following results:
\begin{subequations}	\label{EquIneqStep4}
\begin{align}
& \; \phi(\mathbf{x}^{(k-1)}, \mathbf{y}^{(k-1)}, \mathbf{p}^{(k-1)}) \nonumber \\
= & \;  \xi(\mathbf{x}^{(k-1)}, \mathbf{y}^{(k-1)},\mathbf{u}^{(k-1)}, \mathbf{p}^{(k-1)})  \label{EquIneqStep4a}  \\
= & \; \xi_{lb}(\mathbf{x}^{(k-1)}, \mathbf{y}^{(k-1)},\mathbf{u}^{(k-1)}, \mathbf{p}^{(k-1)})  \label{EquIneqStep4b}  \\
\leq & \; \xi_{lb}(\mathbf{x}^{(k)}, \mathbf{y}^{(k)},\mathbf{u}^{(k)}, \mathbf{p}^{(k-1)})  \label{EquIneqStep4c}  \\
\leq & \; \xi(\mathbf{x}^{(k)}, \mathbf{y}^{(k)},\mathbf{u}^{(k)}, \mathbf{p}^{(k-1)})   \label{EquIneqStep4d}  \\
= & \; \phi(\mathbf{x}^{(k)}, \mathbf{y}^{(k)}, \mathbf{p}^{(k-1)}),   \label{EquIneqStep4e}
\end{align}
\end{subequations}	
where \eqref{EquIneqStep4a} and \eqref{EquIneqStep4e} hold because problems \eqref{EquSubProb2Ori} and \eqref{EquSubProb2Reform} have the same optimal value and optimal solution of $(\mathbf{x},\mathbf{y})$;  \eqref{EquIneqStep4b} holds because the first-order Taylor expansions in \eqref{EquTayloru}, \eqref{EquTaylorx}, and \eqref{EquTaylory} are tight at the feasible point $(\mathbf{x}_{\text{fea}}, \mathbf{y}_{\text{fea}}, \mathbf{u}_{\text{fea}})=(\mathbf{x}^{(k-1)}, \mathbf{y}^{(k-1)}, \mathbf{u}^{(k-1)})$;  \eqref{EquIneqStep4c} holds because $(\mathbf{x}^{(k)}, \mathbf{y}^{(k)}, \mathbf{u}^{(k)})$ is the optimal solution to problem \eqref{EquSubProb2Traject};  \eqref{EquIneqStep4d} holds because the objective value of problem \eqref{EquSubProb2Traject} is a lower bound of that of problem \eqref{EquSubProb2Reform}. Moreover, in step 5, we have the following inequality
\begin{equation} \label{EquIneqStep5}
\phi(\mathbf{x}^{(k)}, \mathbf{y}^{(k)}, \mathbf{p}^{(k-1)}) \leq \phi(\mathbf{x}^{(k)}, \mathbf{y}^{(k)}, \mathbf{p}^{(k)}),
\end{equation}
because $\mathbf{p}^{(k)}$ is the optimal solution to problem \eqref{EquSubProb1}. Based on \eqref{EquIneqStep4} and \eqref{EquIneqStep5}, we obtain
\begin{equation}
\phi(\mathbf{x}^{(k-1)}, \mathbf{y}^{(k-1)}, \mathbf{p}^{(k-1)}) \leq \phi(\mathbf{x}^{(k)}, \mathbf{y}^{(k)}, \mathbf{p}^{(k)}),
\end{equation}
which means that the objective value of problem (P3) is non-decreasing over iterations in Algorithm \ref{AlgDownlink}. In addition, since the optimal value of (P3) is upper-bounded by a finite value, Algorithm \ref{AlgDownlink} is guaranteed to converge. The complexity of Algorithm \ref{AlgDownlink} can be shown to be of $\mathcal{O}(N_{\text{ite}} N^{3.5})$, where $N_{\text{ite}}$ denotes the iteration number.

\section{Proposed Algorithm for Problem (P2)}
In this section, we consider problem (P2) for the G2U case. Similar to (P1), we solve problem (P2) by considering the following equivalent problem,
\begin{align}
\text{(P4)}: & \nonumber \\
\max_{ \mathbf{x}, \mathbf{y}, \mathbf{q} } &  \; \sum_{n=1}^N \bigg[  \log_2 \left( 1+ \frac{ \gamma_0 q[n] } {  ( x[n] - x_{\text{G}} )^2 + ( y[n] - y_{\text{G}} )^2 + H^2 }  \right) \nonumber \\
& \; - \log_2 \left( 1 + \frac{\gamma_0 q[n]}{ d_{\text{GE}}^{\kappa} }  \right)   \bigg] \label{EquP4}  \\
\text{s.t.} &\;   \eqref{EquMobilityCon}, \; \eqref{EquPowerConUL}. \nonumber
\end{align}
Although problem (P4) is non-convex, it has a similar structure as problem (P3), which facilitates us to also apply the block coordinate descent method to find an approximate solution for it. Like (P3), problem (P4) can also be decomposed into two sub-problems: one (denoted by sub-problem 3) is to optimize transmit power $\mathbf{q}$ under given trajectory $(\mathbf{x},\mathbf{y})$; while the other (denoted by sub-problem 4) is to optimize trajectory $(\mathbf{x},\mathbf{y})$ under given transmit power $\mathbf{q}$. The two sub-problems are solved alternately in an iterative manner until convergence. Next, we discuss on how to solve the two sub-problems, respectively.

\subsection{Sub-Problem 3: Optimizing Transmit Power Given UAV Trajectory}  \label{SecPCUL}
With given $(\mathbf{x},\mathbf{y})$, sub-problem 3 can be expressed as
\begin{align}
\max_{\mathbf{q}} \; &  \sum_{n=1}^N  \big[ \log_2 \left( 1+ a_n q[n]   \right)  - \log_2 \left(1+ b q[n] \right) \big] \label{EquSubProb3} \\
\text{s.t.} \; \; & \eqref{EquPowerConUL} ,  \nonumber
\end{align}
where $a_n$ is defined in \eqref{Equa} and
\begin{equation}  \label{Equb2}
b = \frac{ \gamma_0  } { d_{\text{GE}}^{\kappa} }.
\end{equation}
Problem \eqref{EquSubProb3} is similar to sub-problem 1 given in \eqref{EquSubProb1}, thus it can be solved by using \eqref{EquOptPowerScheme} and \eqref{EquOptPowSol}, provided that $b_n$ in \eqref{EquOptPowerScheme} and \eqref{EquOptPowSol} is replaced with $b$ in \eqref{Equb2}.

\subsection{Sub-Problem 4: Optimizing UAV Trajectory Given Transmit Power}
With given $\mathbf{q}$, by letting $Q_n = \gamma_0 q[n]$ and removing the terms in the objective function in \eqref{EquP4} which are irrelevant to $\mathbf{x}$ and $\mathbf{y}$, we express sub-problem 4 as
\begin{align}
\max_{ \mathbf{x}, \mathbf{y} } \; &  \sum_{n=1}^N \log_2 \left( 1+ \frac{ Q_n } {  ( x[n] - x_{\text{G}} )^2 + ( y[n] - y_{\text{G}} )^2 + H^2 }  \right)  \label{EquSubProb4Ori}  \\
\text{s.t.}  \; \; & \eqref{EquMobilityCon}. \nonumber
\end{align}
Unlike sub-problem 2 given in \eqref{EquSubProb2Ori} for the U2G case, problem \eqref{EquSubProb4Ori} is simplified as maximizing only the average achievable rate from the ground node to the UAV. This is because the trajectory of the UAV determines only the channel gain from the ground node to the UAV, but does not have any effect on the channel from the ground node to the eavesdropper.

Despite the non-convexity of problem \eqref{EquSubProb4Ori}, we apply the successive convex optimization to approximately solve it, similar to problem \eqref{EquSubProb2Ori}. First, we introduce slack variable $\mathbf{u} \triangleq [ u[1],\ldots,u[N] ]^\dagger$, and solve the following problem which has the same optimal solution of $(\mathbf{x}, \mathbf{y})$ as problem \eqref{EquSubProb4Ori},
\begin{subequations}   \label{EquSubProb4Reform}
\begin{align}
\max_{ \mathbf{x}, \mathbf{y}, \mathbf{u} } \; & \sum_{n=1}^N \log_2 \left( 1+ \frac{Q_n}{u[n]}  \right)  \label{EquSP4Obj}   \\
\text{s.t.} \; \; & ( x[n] - x_{\text{G}} )^2 + ( y[n] - y_{\text{G}} )^2 + H^2 - u[n] \leq 0, \; \forall n,  \label{EquConU}  \\
& \eqref{EquMobilityCon}. \nonumber
\end{align}
\end{subequations}
With a given initial point $\mathbf{u}_{\text{fea}} \triangleq [ u_{\text{fea}}[1], \ldots, u_{\text{fea}}[N] ]^\dagger$, which is feasible to \eqref{EquSubProb4Reform}, and by applying the first-order Taylor expansion of $\log_2 \big( 1+ \frac{Q_n}{u[n]} \big)$ given in \eqref{EquTayloru} (where $P_n$ is replaced by $Q_n$), problem \eqref{EquSubProb4Reform} is recast as
\begin{align}
\max_{ \mathbf{x}, \mathbf{y}, \mathbf{u} } \; & \sum_{n=1}^N - \frac{Q_n u[n]}{ \ln 2 ( u_{\text{fea}}^2[n] + Q_n u_{\text{fea}}[n] )}   \label{EquSubProb4Recast} \\
\text{s.t.} \; \;& \eqref{EquMobilityCon}, \; \eqref{EquConU}. \nonumber
\end{align}
In can be shown that problem \eqref{EquSubProb4Recast} is a convex quadratically constrained quadratic programming (QCQP) problem, and thus can be efficiently solved by the interior-point method \cite{Boyd2004}. The details of the overall algorithm for solving (P2) are omitted for brevity, given the similarity to that for (P1).

\section{Simulation Results}
In this section, we provide simulation results to verify the performance of our proposed joint UAV trajectory optimization and transmit power control algorithm (denoted as T-OPT-With-PC). For comparison, we also consider the following three benchmark schemes without optimized trajectory and/or power control:
\begin{itemize}
\item Trajectory optimization without transmit power control (denoted as T-OPT-Without-PC);
\item Best-effort trajectory design with transmit power control (denoted as BET-With-PC);
\item Best-effort trajectory design without transmit power control (denoted as BET-Without-PC).
\end{itemize}
Specifically, the T-OPT-Without-PC algorithm designs the UAV trajectories for the U2G and G2U cases by solving problems \eqref{EquSubProb2Traject} and \eqref{EquSubProb4Recast} iteratively until convergence, respectively, with transmit power equally set over time, i.e., $p[n]=\bar{P}$ and $q[n]=\bar{Q}$, $\forall n$. The complexity of the T-OPT-Without-PC algorithm is $\mathcal{O}( N_{\text{ite}} N^{3.5} )$, where $N_{\text{ite}}$ denotes the number of iterations required for convergence \cite{Cui2018}. Both the BET-With-PC and BET-Without-PC algorithms design the UAV trajectories in the following heuristic best-effort manner: the UAV first flies straight to the point above the ground node at its maximum speed, then remains static at that point (if time permits), and finally flies at its maximum speed to reach its final location by the end of the last time slot. Note that if the UAV does not have sufficient time to reach the point above the ground node, it will turn at a certain midway point and then fly to the final location at the maximum speed. Given this trajectory, the BET-With-PC algorithm optimizes the transmit power in the U2G or G2U case by solving problem \eqref{EquSubProb1} or \eqref{EquSubProb3}, respectively; while the BET-Without-PC algorithm sets transmit power equally over time, i.e., $p[n]=\bar{P}$ and $q[n]=\bar{Q}$, $\forall n$. The complexities of the BET-With-PC and BET-Without-PC algorithms are both $\mathcal{O}( N )$. The initial feasible solutions for the proposed T-OPT-With-PC and benchmark T-OPT-Without-PC algorithms are generated by the BET-Without-PC algorithm.

The coordinates of the ground node and the eavesdropper are set as $(0,0,0)$m and $(200, 0, 0)$m, respectively, and the flying altitude of the UAV is set as $H=100$m. The maximum speed of the UAV is $v_{\max}=3$m/s. The flight period $T$ is divided into multiple time slots with equal length of $d_t=0.5$s. The communication bandwidth is $20$MHz with the carrier frequency at $5$GHz, and the noise power spectrum density is $-169$dBm/Hz. Thus, the reference SNR at the reference distance $d_0=1$m is $\gamma_0=80$dB. The peak transmit power limits for the U2G and G2U links are set as $P_{\text{peak}} =4 \bar{P}$ and $Q_{\text{peak}} =4 \bar{Q}$, respectively. The threshold in Algorithm \ref{AlgDownlink} is set as $\epsilon=10^{-4}$. We consider two cases, denoted as Case 1 and Case 2, in which the UAV has different initial and final locations. In Case 1, the UAV's initial and final locations are on the vertical line in the middle of the ground node and the eavesdropper, and the coordinates of them are set as $(x_0,y_0)=(100,600)$m and $(x_F,y_F)=(100,-600)$m, as shown in Fig. \ref{FigTraj_DL_Ver}. In Case 2, the UAV's initial and final locations are on the parallel line of that connecting the ground node and the eavesdropper, and the coordinates of them are set as $(x_0,y_0)=(-500,-150)$m and $(x_F,y_F)=(700,-150)$m, as shown in Figs. \ref{FigTraj_DL_Para} and \ref{FigTraj_UL_Para}.

\subsection{U2G Communication}
\begin{figure}[!t]
	\centering
	\includegraphics[width=\columnwidth]{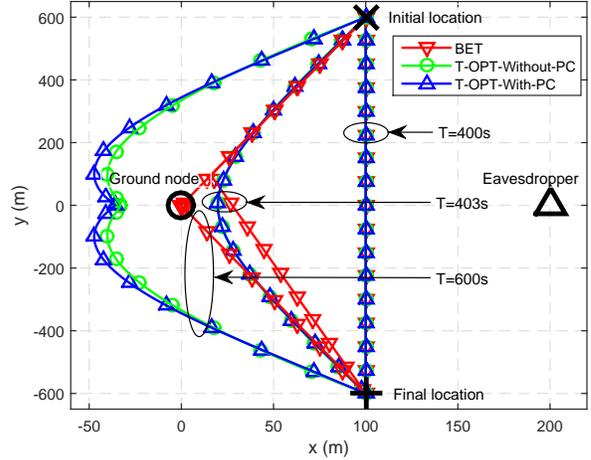}
	\caption{Trajectories of the UAV for the U2G communication in Case 1.}  \label{FigTraj_DL_Ver}
\end{figure}

In the U2G case, we first consider Case 1. Fig. \ref{FigTraj_DL_Ver} shows the trajectories of the UAV by applying different algorithms with different values of flight period $T$. The average transmit power is set as $\bar{P}=-5$dBm. The locations of the ground node, eavesdropper, as well as the UAV's initial and final locations are marked with $\bigcirc$, $\triangle$, $\times$, and $+$, respectively. It is observed that when $T=400$s which is the minimum time required for the UAV to fly from the initial location to the final location in a straight line at the maximum speed, the trajectories of all algorithms are identical. As $T$ increases, the trajectories by the proposed T-OPT-With-PC and the benchmark T-OPT-Without-PC algorithms are still similar, i.e., the UAV tries to fly as close as possible to the ground node and at the same time as far away as possible to the eavesdropper, while they appear different compared to that by the heuristic best-effort trajectory (BET) design (same for with and without power control). When $T$ is sufficiently large, i.e. $T = 600$s, for the proposed T-OPT-With-PC algorithm or the benchmark T-OPT-Without-PC algorithm, it is observed that the UAV first flies at the maximum speed to reach a certain location on the left of the ground node (not directly above it), then remains stationary at this location as long as possible, and finally flies to the final location in an arc path at the maximum speed and reaches there by the end of the last time slot. These stationary locations, which are on the opposite direction of the eavesdropper, strike an optimal balance between enhancing the legitimate link channel and degrading the eavesdropping link channel and hence maximize the secrecy rate in each of these two cases. In fact, this is also why the UAV follows an arc path rather than the straight path as in the heuristic BET design.

\begin{figure}[!t]
	\centering
	\includegraphics[width=\columnwidth]{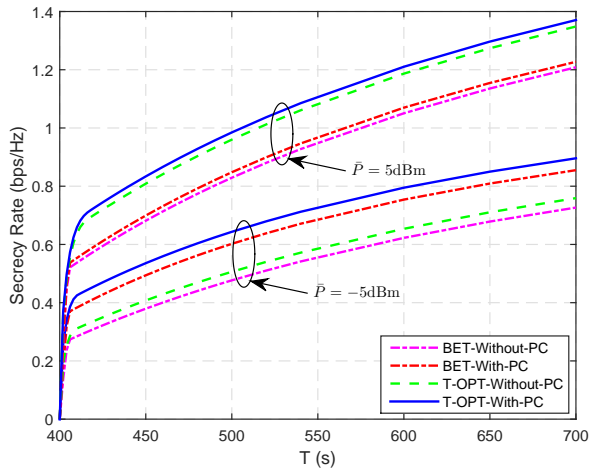}
	\caption{Secrecy rate versus flight period $T$ for the U2G communication in Case 1.}  \label{FigSRvsT_DL_Ver}
\end{figure}

\begin{figure}
	\centering	\includegraphics[width=\columnwidth]{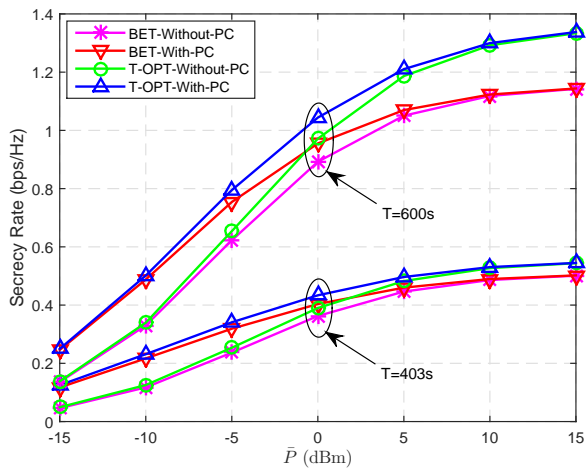}
	\caption{Secrecy rate versus average transmit power $\bar{P}$ for the U2G communication in Case 1.}  \label{FigSRvsP_DL_Ver}
\end{figure}

Fig. \ref{FigSRvsT_DL_Ver} shows the corresponding average secrecy rates of different algorithms versus flight period $T$ when $\bar{P}=-5$dBm and $5$dBm. It is observed that the secrecy rates of all algorithms increase significantly with $T$. This is because for all algorithms the maximum secrecy rate is achieved at their respective stationary locations (see Fig. \ref{FigTraj_DL_Ver}) for sufficiently large $T$, and larger $T$ results in longer hovering time at such locations for the UAV. The proposed T-OPT-With-PC algorithm always achieves the highest secrecy rate, while the benchmark BET-Without-PC algorithm has the lowest secrecy rate, as expected. When $\bar{P}=-5$dBm, the benchmark BET-With-PC algorithm has higher secrecy rate than that of the benchmark T-OPT-Without-PC algorithm. In contrast, when $\bar{P}=5$dBm, the latter algorithm has higher secrecy rate than the former. Such results suggest that in the low transmit power regime, power control is more important for improving the secrecy rate; while in the high transmit power regime, trajectory optimization is more significant. Furthermore, it is worth pointing out that trajectory adaptation with increasing $T$ is essential for the secrecy rate improvement, even for the case with heuristic BET design. Otherwise, if the trajectory is fixed (e.g., following the straight line from the initial to final location with constant speed of $\sqrt{(x_F -x_0)^2 + (y_F - y_0)^2}/T$ regardless of $T$, then the secrecy rate will remain unchanged as the case with required minimum $T=400$s in Fig. \ref{FigSRvsT_DL_Ver}, i.e., the UAV cannot exploit its mobility to improve the secrecy rate, even with sufficiently large $T$.

Fig. \ref{FigSRvsP_DL_Ver} shows the average secrecy rates of different algorithms versus the average transmit power $\bar{P}$ when $T=403$s and $600$s. It is observed that the proposed T-OPT-With-PC algorithm always achieves the highest secrecy rate, while the benchmark BET-Without-PC algorithm has the lowest secrecy rate. In addition, when $\bar{P} \leq -5$dB, we note that the benchmark BET-With-PC algorithm achieves a secrecy rate performance close to the proposed T-OPT-With-PC algorithm, and also significantly outperforms the benchmark T-OPT-Without-PC algorithm. However, when $\bar{P}$ increases, the secrecy rate of the benchmark T-OPT-Without-PC algorithm will eventually exceed that of the benchmark BET-With-PC algorithm and even get closer to that of the proposed T-OPT-With-PC algorithm. Furthermore, the rate gap between the benchmark T-OPT-Without-PC and BET-With-PC algorithms becomes larger with increasing $\bar{P}$. The above results further demonstrate that transmit power control is more effective in improving secrecy rate than trajectory optimization when the average transmit power is low, while trajectory optimization is more effective when the average transmit power is high.

\begin{figure}[!t]
	\centering
	\includegraphics[width=\columnwidth]{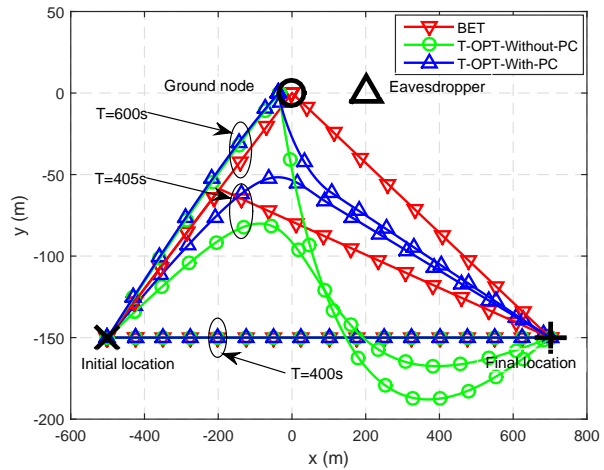}
	\caption{Trajectories of the UAV for the U2G communication in Case 2.}  \label{FigTraj_DL_Para}
\end{figure}

\begin{figure}
	\centering
	\includegraphics[width=\columnwidth]{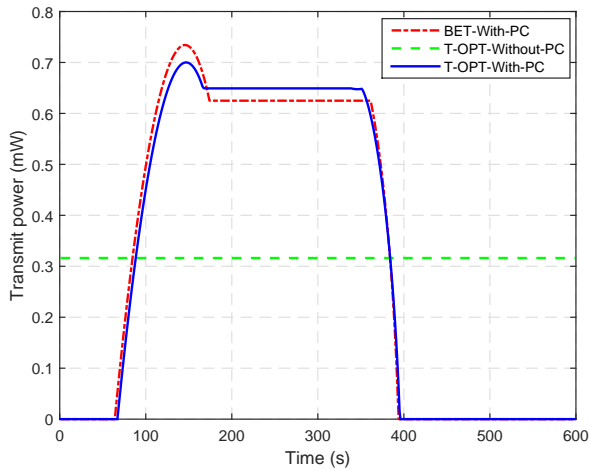}
	\caption{Transmit power versus time slot for the U2G communication in Case 2 when $T=600$s.}  \label{FigPowvsTime_DL_Para}
\end{figure}

\begin{figure}[!t]
	\centering
	\includegraphics[width=\columnwidth]{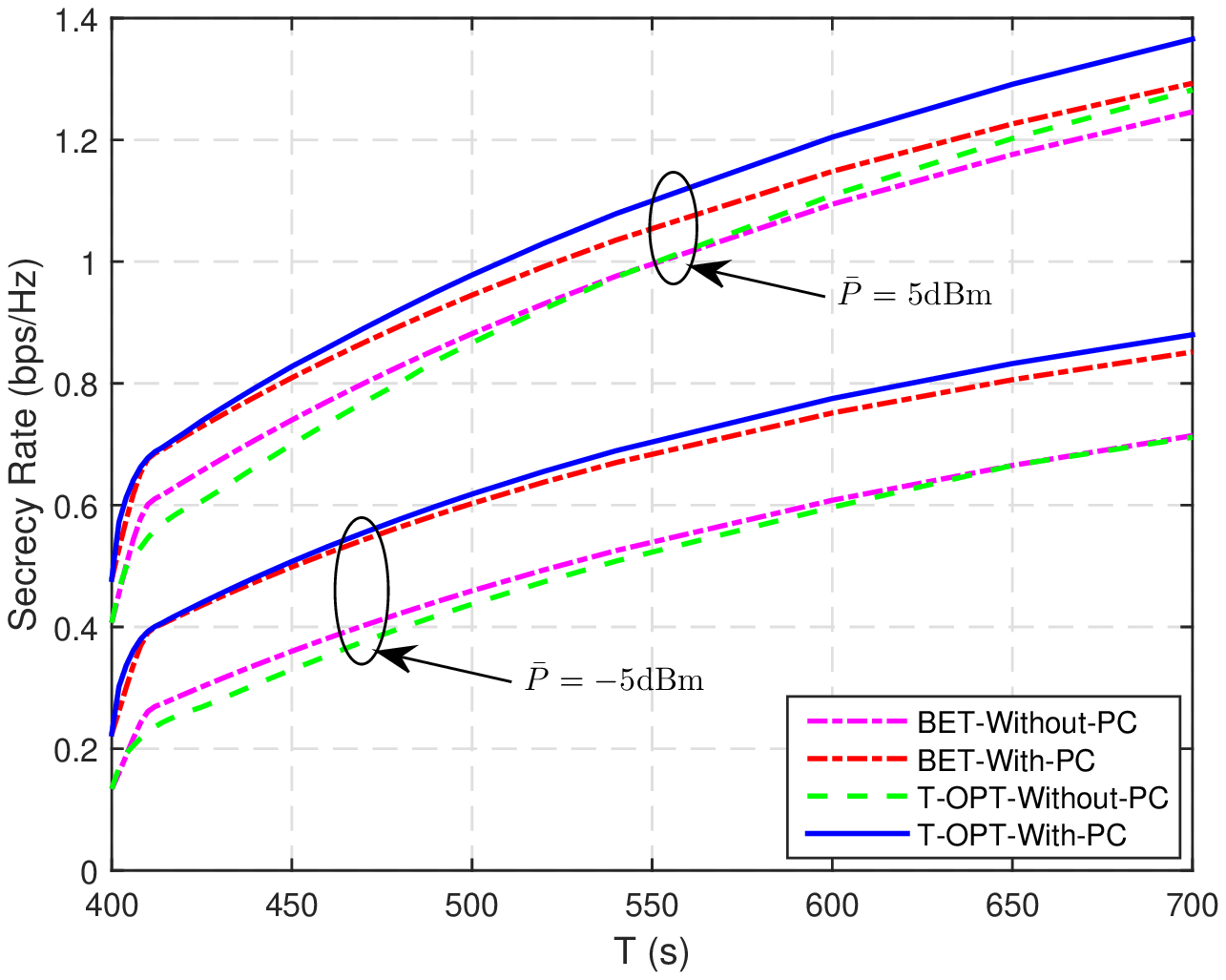}
	\caption{Secrecy rate versus flight period $T$ for the U2G communication in Case 2.}  \label{FigSRvsT_DL_Para}
\end{figure}

\begin{figure}[!t]
	\centering
	\includegraphics[width=\columnwidth]{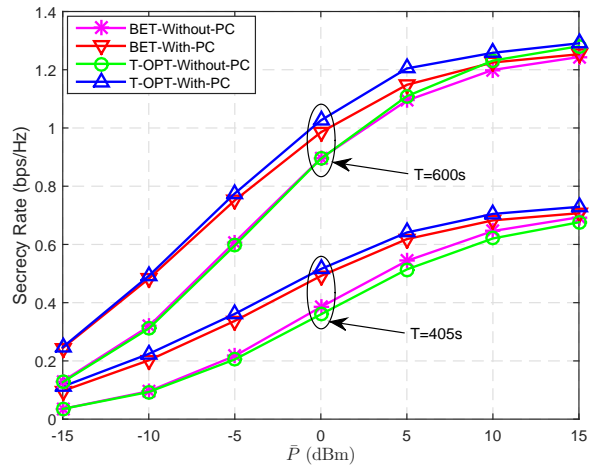}
	\caption{Secrecy rate versus average transmit power $\bar{P}$ for the U2G communication in Case 2.}  \label{FigSRvsP_DL_Para}
\end{figure}

Next, we consider Case 2. Fig. \ref{FigTraj_DL_Para} shows the trajectories of the UAV by using different algorithms when $\bar{P}=-5$dBm. It is observed that different from the results in Case 1 shown in Fig. \ref{FigTraj_DL_Ver}, the trajectories by the proposed T-OPT-With-PC and benchmark T-OPT-Without-PC algorithms with $T=405$s or $600$s differ significantly, especially when the UAV flies towards the final location. For the proposed algorithm, the UAV flies along a relatively direct path towards the ground node and then towards the final location. In contrast, for this benchmark, the UAV first flies almost directly to the ground node, but then flies along an arc path to the final location, which inevitably consumes more time on the traveling compared to the trajectory of the proposed algorithm. The reason for such a difference is that in Case 2, flying from the ground node towards the final location reduces the distance from the UAV to the eavesdropper less much as compared to that from it to the ground node, which is undesired. This means that to improve the secrecy rate, the UAV should reduce transmit power when it gets farther away from the ground node and closer to the final location. Considering this fact, the proposed T-OPT-With-PC algorithm is able to decrease the UAV transmit power or even turn off the transmitter to save power and also protect from eavesdropping when the UAV flies directly towards the final location. However, for the benchmark T-OPT-Without-PC algorithm that employs a constant transmit power, the UAV can only rely on adjusting its trajectory to keep far away from the eavesdropper to avoid being eavesdropped, which however requires more traveling time and leads to a longer arc trajectory. This fact is verified by Fig. \ref{FigPowvsTime_DL_Para}, which shows the transmit power of the UAV over time slot when the flight period is $T=600$s. It is observed that both the proposed T-OPT-With-PC and benchmark BET-With-PC algorithms increase UAV transmit power when the UAV gets closer to the ground node, and reduce UAV transmit power when the UAV gets farther away from the ground node and closer to the final location. When the UAV is in the zone where the distance from it to the ground node is larger than that to the eavesdropper, the proposed T-OPT-With-PC and benchmark BET-With-PC algorithms set UAV transmit power to zero.

With the UAV trajectory difference in Fig. \ref{FigTraj_DL_Para}, the secrecy rate performances of different algorithms versus $T$ and $\bar{P}$, which are shown in Figs. \ref{FigSRvsT_DL_Para} and \ref{FigSRvsP_DL_Para} respectively, are also quite different from those shown in Figs. \ref{FigSRvsT_DL_Ver} and \ref{FigSRvsP_DL_Ver} for Case 1. Specifically, the secrecy rate gaps between the proposed T-OPT-With-PC and benchmark T-OPT-Without-PC algorithms versus $T$ or $\bar{P}$ in Case 2 are significantly larger than those in Case 1. For example, in Fig. \ref{FigSRvsT_DL_Para}, the T-OPT-Without-PC algorithm even has lower secrecy rate than the BET-Without-PC algorithm in the regime of $T \leq 650$s when $\bar{P}=-5$dBm or in the regime of $T\leq 550$s when $\bar{P}=5$dBm. In addition, in Fig. \ref{FigSRvsP_DL_Para}, the T-OPT-Without-PC algorithm has a lower secrecy rate than the BET-Without-PC algorithm over the whole $\bar{P}$ regime when $T=405$s and in the regime of $\bar{P} \leq 0$dBm when $T=600$s. This is mainly because the UAV wastes more time on travelling along a longer arc trajectory to reach the final location which in turn leads to the inefficient use of the transmit power. The above results demonstrate the importance and necessity of the joint UAV trajectory optimization and transmit power control in maximizing the secrecy rate for U2G communication.

From Figs. \ref{FigTraj_DL_Ver}--\ref{FigSRvsP_DL_Para}, it is observed that although the proposed T-OPT-With-PC algorithm always achieves the highest secrecy rate in all cases, other benchmark algorithms of lower complexity may achieve reasonably good performance as compared to the proposed algorithm in certain cases. As such, depending on the system parameters (e.g., average transmit power, flight period, and the UAV's initial and final locations), the UAV may adopt different algorithms to strike a balance between the achievable performance and computational complexity. In particular, for the case with short $T$ and/or high $\bar{P}$, trajectory optimization is generally less effective as compared to transmit power control in improving secrecy rate, thus the benchmark algorithm BET-With-PC performs very close to the proposed algorithm; while if both the initial and final locations of the UAV are closer to the ground node than the eavesdropper, trajectory optimization is more effective, thus T-OPT-Without-PC is a good choice from both performance and complexity considerations.

\subsection{G2U Communication}

In the G2U case, the channel gain from the ground node to the eavesdropper given in \eqref{EquChGain} contains a small-scale Rayleigh fading term $\zeta$. Thus, all secrecy rate results in the following are averaged over $5000$ random independent realizations of $\zeta$, where the path-loss exponent is set as $\kappa=3$. Since our results obtained for Cases 1 and 2 lead to consistent observations, we only present the results for Case 2 due to the space limitation.

Fig. \ref{FigTraj_UL_Para} shows the trajectories of the UAV with different values of $T$ when the average transmit power is $\bar{Q}=-5$dBm. It is observed that the trajectories of the proposed T-OPT-With-PC and benchmark T-OPT-Without-PC algorithms are very similar for different values of $T$, i.e., the UAV tries to fly as close as possible to the ground node as $T$ increases. When $T$ is sufficiently large, i.e., $T=600$s, the trajectories of them are the same as the heuristic BET design, i.e., the UAV first flies at the maximum speed to reach the point right above the ground node, then remains static as long as possible, and finally flies to the final location directly at the maximum speed in order to reach there by the end of the last time slot. The fundamental reason of such a result is that in the G2U setup, the channel between the ground node (transmitter) and the eavesdropper is independent of the UAV's location. Therefore, the UAV trajectory is only optimized to maximize achievable rate from the ground node to the UAV. Obviously, the point right above the ground node is the best location for achieving its largest rate. This explains why the optimized trajectory also converges to the heuristic BET design when $T$ is sufficiently large.

\begin{figure}[!t]
	\centering	\includegraphics[width=\columnwidth]{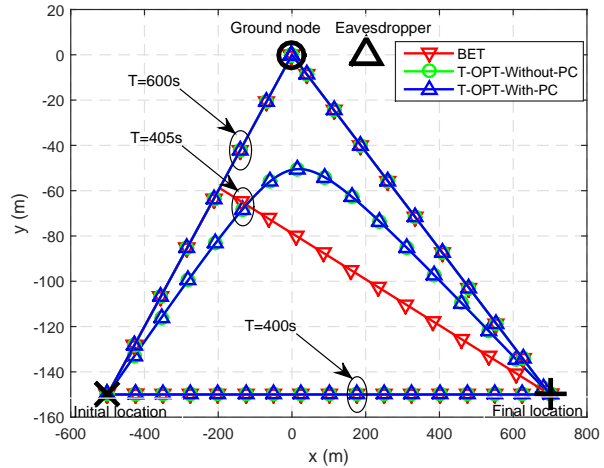}
	\caption{Trajectories of the UAV for the G2U communication in Case 2.}  \label{FigTraj_UL_Para}
\end{figure}

Fig. \ref{FigSRvsT_UL_Para} shows the average secrecy rates of different algorithms versus flight period $T$ when $\bar{Q}=-5$dBm and $5$dBm. When $T \geq 410$s, the algorithms with transmit power control, i.e. the proposed T-OPT-With-PC and benchmark BET-With-PC algorithms, achieve the same secrecy rate since they have the same trajectory and hence the same transmit power control, while they both outperform the benchmark algorithms without power control, i.e., the T-OPT-Without-PC and BET-Without-PC algorithms. These results suggest that transmit power control is more effective than trajectory optimization in improving secrecy rate in the G2U case, as shown in Fig. \ref{FigTraj_UL_Para}, since the optimal trajectory can be easily achieved by the BET design when $T$ is sufficiently large. Furthermore, the secrecy rate gap between the algorithms with and without power control when $\bar{Q}=-5$dBm is significantly larger than that when $\bar{Q}=5$dBm, since power control is more effective when the average transmit power is low.

\begin{figure}[!t]
	\centering	\includegraphics[width=\columnwidth]{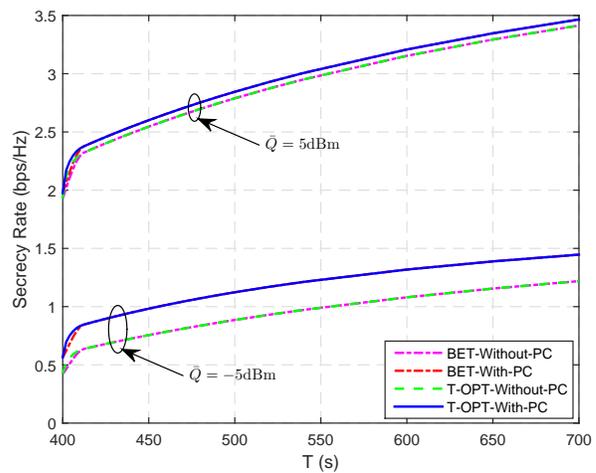}
	\caption{Secrecy rate versus flight period $T$ for the G2U communication in Case 2.}  \label{FigSRvsT_UL_Para}
\end{figure}

\begin{figure}[!t]
	\centering	\includegraphics[width=\columnwidth]{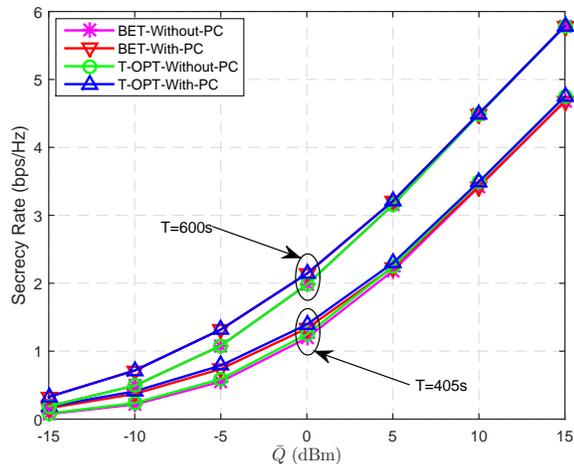}
	\caption{Secrecy rate versus average transmit power $\bar{Q}$ for the G2U communication in Case 2.}  \label{FigSRvsP_UL_Para}
\end{figure}

Fig. \ref{FigSRvsP_UL_Para} shows the average secrecy rates of different algorithms versus average transmit power $\bar{Q}$ when $T=405$s and $600$s. It can be also observed that transmit power control is effective for improving secrecy rate when $\bar{Q} \leq 0$dBm. When $T=405$s, the secrecy rate gap between the proposed T-OPT-With-PC algorithm and the benchmark algorithms with BET design exists due to their trajectory difference. When $T=600$s, the secrecy rates of all algorithms tend to be very similar when $\bar{Q} \geq 10$dBm. This is because their trajectories are the same and the power control only provides marginal rate gain when transmit power is high.

\section{Conclusion}
In this paper, we study the physical layer security for emerging UAV communications in the forthcoming 5G wireless networks. Specifically, we propose to enhance the security performance by proactively controlling channel gains via adjusting the UAV trajectory in addition to applying the conventional power/rate adaptation, which leads to a new joint optimization framework. For both the U2G and G2U communications, the transmit power control and UAV trajectory are jointly designed to maximize the average secrecy rate over a finite horizon, subject to the average and peak transmit power constraints as well as practical UAV's mobility constraints. By applying the block coordinate descent and successive convex optimization methods, efficient iterative algorithms are proposed to solve the joint design problems. Simulation results show that joint trajectory optimization and transmit power control improves the physical layer security performance, and more significantly in the U2G case compared to the G2U case, as the UAV trajectory in the U2G case has an effect on both the legitimate and eavesdropping channels, instead of the legitimate channel only in the G2U case. Furthermore, it is found that both UAV trajectory optimization and transmit power control are generally necessary in the U2G case; while in the G2U case, transmit power control is more effective than UAV trajectory optimization for improving the secrecy rate performance, and the heuristic best-effort trajectory already performs quite close to the optimized trajectory.

\begin{IEEEbiography}[{\includegraphics[width=1in,height=1.25in,clip,keepaspectratio]{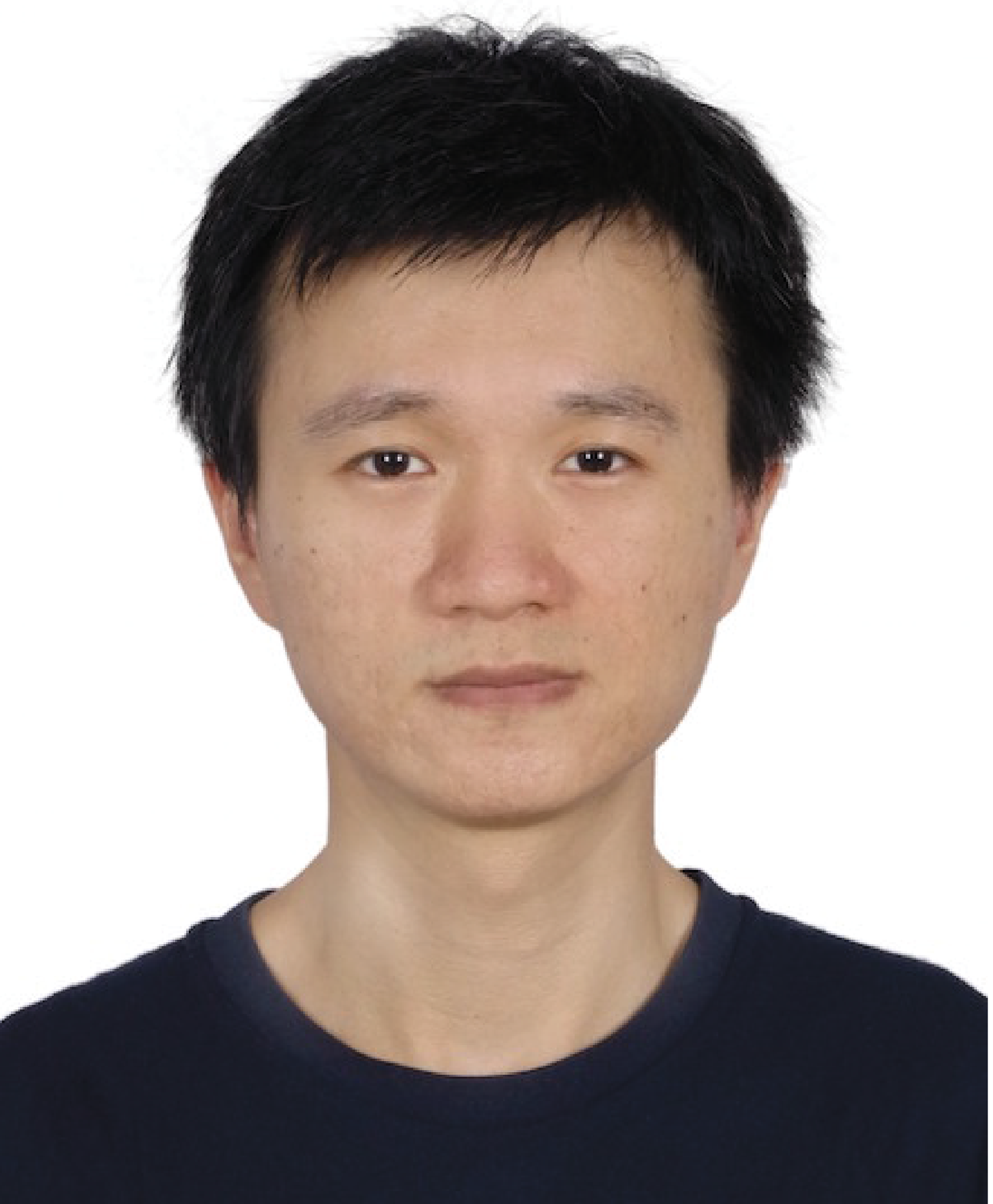}}]{Guangchi Zhang} (M'13) received the B.S. degree in electronic engineering from the Nanjing University, Nanjing, China, in 2004, and the Ph.D. degree in communication engineering from the Sun Yat-Sen University, Guangzhou, China, in 2009. He has been with the Guangdong University of Technology since 2009. He was a Senior Research Associate with the City University of Hong Kong from Oct. 2011 to Mar. 2012 and a Visiting Professor with the National University of Singapore from Jan. 2017 to Jan. 2018. He is currently a Professor with the School of Information Engineering, Guangdong University of Technology, Guangzhou, China. His research interests include MIMO and relay wireless communications, wireless power transfer, unmanned aerial vehicle communications, and physical layer security. He was a recipient of the IEEE Communications Society 2014 Heinrich Hertz Award and the IEEE Communication Letters 2014 Exemplary Reviewer.
\end{IEEEbiography}

\begin{IEEEbiography}
[{\includegraphics[width=1in,height=1.25in,clip,keepaspectratio]{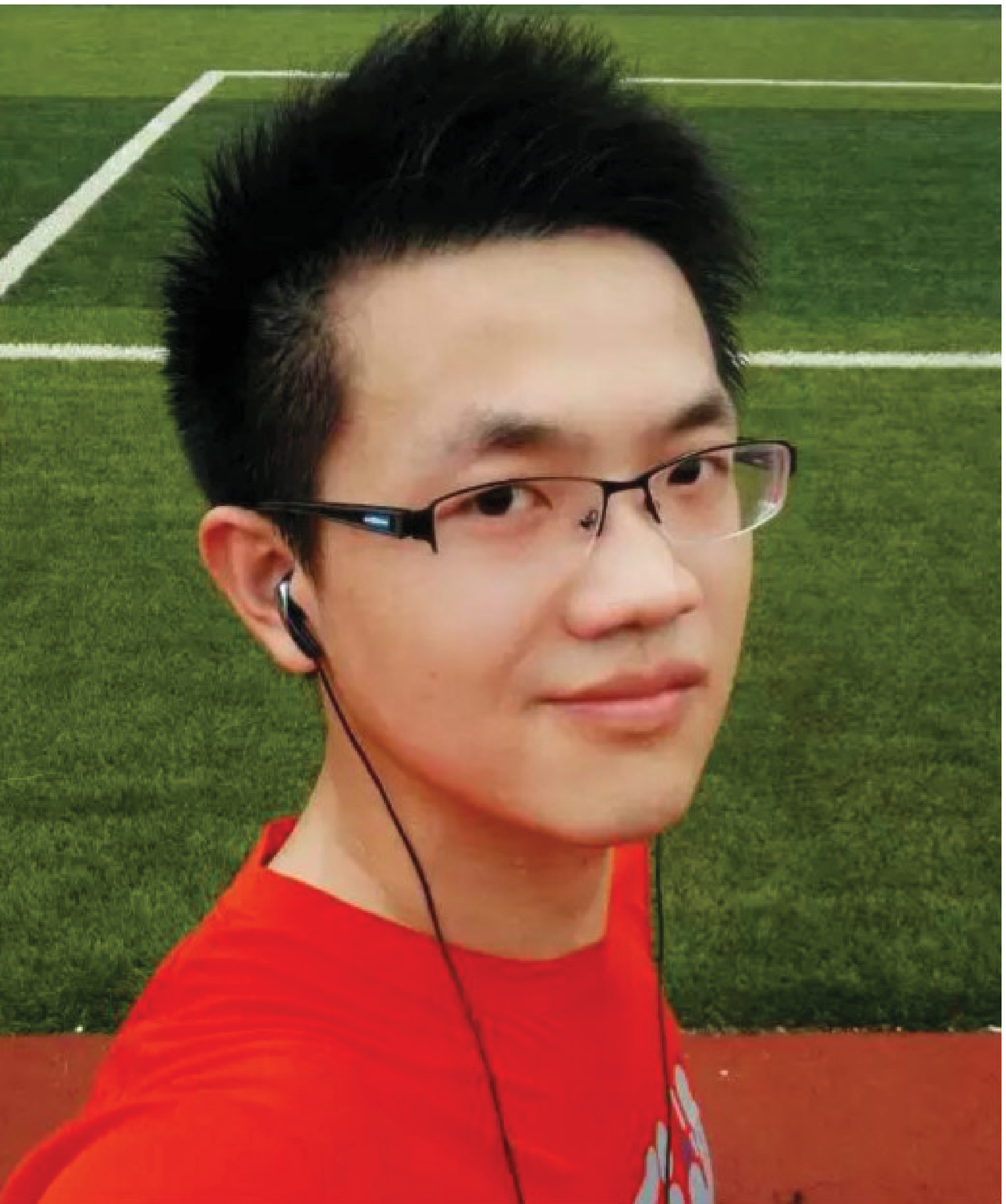}}]{Qingqing Wu} (S'13-M'16) received B.Eng. and the Ph.D. degrees in Electronic Engineering from South China University of Technology and Shanghai Jiao Tong University (SJTU), China, in 2012 and 2016 (in advance), respectively. He is now a Research Fellow in National University of Singapore. He received the IEEE WCSP Best Paper Award in 2015, the Exemplary Reviewer of IEEE Communications Letters in 2016 and 2017, and the Exemplary Reviewer of IEEE Transactions on Communications and IEEE Transactions on Wireless Communications in 2017. He was the recipient of the Outstanding Ph.D. Thesis Funding in SJTU in 2016 and the Best Ph.D. Thesis Award of China Institute of Communications in 2017. He served as a TPC member of IEEE ICC, GLOBECOM, WCNC, VTC, APCC, WCSP, etc. He is currently an Editor of IEEE Communications
Letters and the workshop co-chair of ICC 2019. His research interests include intelligent reflecting surface (IRS), energy-efficient wireless communications, wireless power transfer, and unmanned aerial vehicle (UAV) communications.
\end{IEEEbiography}

\begin{IEEEbiography}
[{\includegraphics[width=1in,height=1.25in,clip,keepaspectratio]{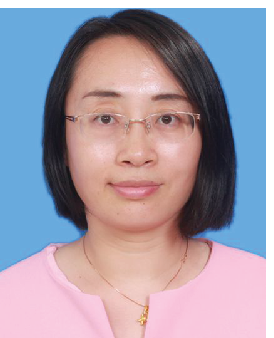}}]{Miao Cui} received the B.E. degree in communication engineering and the M.S. degree in computer science from the Northeast Electric Power University, Jilin, China, in 2001 and 2003, respectively, and the Ph.D. degree in circuit system from the South China University of Technology, Guangzhou, China, in 2009. She is currently a Lecturer with the Guangdong University of Technology, Guangzhou, China. Her research interests include the analysis, optimization, and design of wireless networks.
\end{IEEEbiography}

\begin{IEEEbiography}
[{\includegraphics[width=1in,height=1.25in,clip,keepaspectratio]{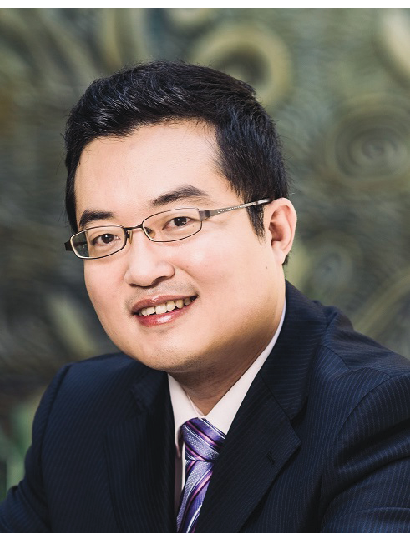}}]{Rui Zhang}  (S'00-M'07-SM'15-F'17) received the B.Eng. (first-class Hons.) and M.Eng. degrees from the National University of Singapore, Singapore, and the Ph.D. degree from the Stanford University, Stanford, CA, USA, all in electrical engineering.

From 2007 to 2010, he worked as a Research Scientist with the Institute for Infocomm Research, ASTAR, Singapore. Since 2010, he has joined the Department of Electrical and Computer Engineering, National University of Singapore, where he is now a Dean's Chair Associate Professor in the Faculty of Engineering. He has authored over 300 papers. He has been listed as a Highly Cited Researcher (also known as the World's Most Influential Scientific Minds), by Thomson Reuters (Clarivate Analytics) since 2015. His research interests include UAV/satellite communication, wireless information and power transfer, multiuser MIMO, smart and reconfigurable environment, and optimization methods.

He was the recipient of the 6th IEEE Communications Society Asia-Pacific Region Best Young Researcher Award in 2011, and the Young Researcher Award of National University of Singapore in 2015. He was the co-recipient of the IEEE Marconi Prize Paper Award in Wireless Communications in 2015, the IEEE Communications Society Asia-Pacific Region Best Paper Award in 2016, the IEEE Signal Processing Society Best Paper Award in 2016, the IEEE Communications Society Heinrich Hertz Prize Paper Award in 2017, the IEEE Signal Processing Society Donald G. Fink Overview Paper Award in 2017, and the IEEE Technical Committee on Green Communications  \& Computing (TCGCC) Best Journal Paper Award in 2017. His co-authored paper received the IEEE Signal Processing Society Young Author Best Paper Award in 2017. He served for over 30 international conferences as the TPC co-chair or an organizing committee member, and as the guest editor for 3 special issues in the IEEE JOURNAL OF SELECTED TOPICS IN SIGNAL PROCESSING and the IEEE JOURNAL ON SELECTED AREAS IN COMMUNICATIONS. He was an elected member of the IEEE Signal Processing Society SPCOM Technical Committee from 2012 to 2017 and SAM Technical Committee from 2013 to 2015, and served as the Vice Chair of the IEEE Communications Society Asia-Pacific Board Technical Affairs Committee from 2014 to 2015. He served as an Editor for the IEEE TRANSACTIONS ON WIRELESS COMMUNICATIONS from 2012 to 2016, the IEEE JOURNAL ON SELECTED AREAS IN COMMUNICATIONS: Green Communications and Networking Series from 2015 to 2016, and the IEEE TRANSACTIONS ON SIGNAL PROCESSING from 2013 to 2017. He is now an Editor for the IEEE TRANSACTIONS ON COMMUNICATIONS and the IEEE TRANSACTIONS ON GREEN COMMUNICATIONS AND NETWORKING. He serves as a member of the Steering Committee of the IEEE Wireless Communications Letters. He is an IEEE Signal Processing Society Distinguished Lecturer.
\end{IEEEbiography}

\end{document}